\documentclass[twocolumn]{IEEEtran}
\usepackage{amsmath,amsfonts}
\usepackage{algorithmic}
\usepackage{algorithm}
\usepackage{array}
\usepackage[caption=false,font=normalsize,labelfont=sf,textfont=sf]{subfig}
\usepackage{textcomp}
\usepackage{stfloats}
\usepackage{url}
\usepackage{verbatim}
\usepackage{graphicx}
\usepackage{cite}
\usepackage{amssymb}
\usepackage{graphicx}
\usepackage{textcomp}
\usepackage{xcolor}
\def\BibTeX{{\rm B\kern-.05em{\sc i\kern-.025em b}\kern-.08em
    T\kern-.1667em\lower.7ex\hbox{E}\kern-.125emX}}
\usepackage{color,epsfig,ifthen,twoopt,tabularx,xargs}
\usepackage{amsthm}
\usepackage[boxruled,algo2e]{algorithm2e}
\usepackage{enumitem,verbatim}
\usepackage{amsfonts}
\usepackage[bbgreekl]{mathbbol}
\usepackage{mathrsfs}
\usepackage{stmaryrd}
\usepackage{pgfplotstable}

\DeclareSymbolFontAlphabet{\mathbbm}{bbold}
\DeclareSymbolFontAlphabet{\mathbb}{AMSb}%
\PassOptionsToPackage{hyphens}{url}\usepackage{hyperref}
\hypersetup{colorlinks=true,citecolor=blue, urlcolor=blue, linkcolor=black, filecolor=black}
\usepackage{cleveref}
\newtheorem{theorem}{Theorem}[section]
\crefname{theorem}{theorem}{theorems}
\Crefname{Theorem}{Theorem}{Definitions}

\newtheorem{definition}{Definition}[section]
\crefname{definition}{definition}{definitions}
\Crefname{Definition}{Definition}{Definitions}

\crefname{corollary}{corollary}{corollaries}
\Crefname{Corollary}{Corollary}{Corollaries}

\newtheorem{proposition}[theorem]{Proposition}
\crefname{proposition}{proposition}{propositions}
\Crefname{Proposition}{Proposition}{Propositions}

\newtheorem{lemma}[theorem]{Lemma}
\crefname{lemma}{lemma}{lemmas}
\Crefname{Lemma}{Lemma}{Lemmas}

\crefname{example}{example}{examples}
\Crefname{example}{Example}{Examples}

\crefname{remark}{remark}{remarks}
\Crefname{remark}{Remark}{Remarks}

\crefname{assumption}{assumption}{assumptions}
\Crefname{assumption}{Assumption}{Assumptions}

\crefname{enumi}{point}{point}
\Crefname{enumi}{Point}{Point}

\theoremstyle{remark}
\newcommand{\norm}[1]{{
    \left\| #1 \right\|}} % norm
\newcommand{\osego}[1]{{\left( #1 \right)}} % open segment ] [
\newcommand{\oseg}[1]{{\left( #1 \right]}} % semi-open segment ] ]
\newcommand{\iseg}[1]{{\left\llbracket #1 \right\rrbracket}} % integer segment

\newcommandx\cspanarg[2][1=]{\ensuremath{\overline{\mathrm{Span}}^{#1}\left(#2\right)}}

\newcommand{\varsqrt}[1]{\left(#1\right)^{1/2}}
\newcommand{\spec}{{\rm \spec}}

\newcommand{\tnorm}[1]{{\left\vert\kern-0.25ex\left\vert\kern-0.25ex\left\vert #1 
    \right\vert\kern-0.25ex\right\vert\kern-0.25ex\right\vert}} % triple norm

\def\rset{\mathbb{R}}

\def\zset{\mathbb{Z}}
\def\nset{\mathbb{N}}

\usepackage{amsfonts}
\usepackage[bbgreekl]{mathbbol}
\DeclareSymbolFontAlphabet{\mathbbm}{bbold}
\DeclareSymbolFontAlphabet{\mathbb}{AMSb}%
  {\begin{enumerate}}%
  {\end{enumerate}}

\def\rmi{\mathrm{i}}
\def\rme{\mathrm{e}}
\def\rmd{\mathrm{d}}

\newcommandx{\aslim}[1]{\ensuremath{\stackrel{#1\text{a.s.}}{\longrightarrow}}}  % david removed dash - before \text
\newcommand{\1}{\mathbbm{1}}

\def\bA{\mathbf{A}}
\def\ba{\mathbf{a}}

\def\bE{\mathbf{E}}

\def\bw{\mathbf{w}}
\def\bW{\mathbf{W}}

\def\bX{\mathbf{X}}

\def\bY{\mathbf{Y}}

\def\eqsp{\;}

\def\balpha{\boldsymbol{\alpha}}

\def\btheta{\boldsymbol{\theta}}
\def\blambda{\boldsymbol{\lambda}}

\def\blambda{\boldsymbol{\lambda}}

\def\cI{\mathcal{I}}

\newcommand{\pscal}[2]{\left\langle #1, #2 \right\rangle}

\def\bfrho{{\boldsymbol{\rho}}}

\def\b0{{\bf 0}}
\def\ba{{\bf a}}

\def\bw{{\bf w}}

\def\bA{{\bf A}}

\def\bE{{\bf E}}

\def\bW{{\bf W}}
\def\bX{{\bf X}}
\def\bY{{\bf Y}}

\def\cA{\mathcal{A}}

\def\cC{\mathcal{C}}

\def\cI{\mathcal{I}}
\def\cJ{\mathcal{J}}

\def\cN{\mathcal{N}}

\def\cP{\mathcal{P}}

\def\cX{\mathcal{X}}

\newcommand{\argmin}{\mathop{\mathrm{argmin}}}

\newcommand\algo[1]%
{
    \begin{center} %
    \begin{tabular} {||p{10 cm}l ||}%
               #1 &  \\
    \hline
    \end{tabular}
    \end{center}
}
\newcommand{\Y}{\ensuremath{Y}}

\newcommand{\chunk}[4][]%
{\ifthenelse{\equal{#1}{}}{\ensuremath{{#2}_{#3:#4}}}{\ensuremath{#2^#1}_{#3:#4}}
}

\def\esp{\mathbb{E}}
\newcommandx\prob[2][1=,2=]{\ensuremath{{\mathbb P}_{#1}^{#2}}}
\newcommand{\PP}[1][]{\ifthenelse{\equal{#1}{}}{\ensuremath{\mathbb{P}}}{\ensuremath{\mathbb{P}\left( #1 \right)}}}
\newcommandx{\PParg}[2][1=]{\PP_{#1}\left(#2\right)}
\newcommand{\PE}[1][]{\ifthenelse{\equal{#1}{}}{\esp}{\ensuremath{{\mathbb E}\left[ #1 \right]}}}
\newcommandx{\PEarg}[2][1=]{\PE_{#1}\left[#2\right]}
\newcommand{\PVar}{\ensuremath{\operatorname{Var}}}
\newcommandx\var[2][1=]{\ensuremath{\PVar_{#1}\left( #2\right)}}
\newcommandx\cvar[3][1=]{\ensuremath{\PVar_{#1}\left( \left. #2 \right| #3 \right)}}

\newcommandx\cov[3][1=]{\ensuremath{\mathrm{Cov}_{#1}\left( #2,#3 \right)}}
\newcommandx\ccov[3][1=]{\ensuremath{\mathrm{Cov}_{#1}\left( \left. #2 \right| #3 \right)}}

\newcommand{\CPP}[3][]
{\ifthenelse{\equal{#1}{}}{\PP\left[\left. #2 \, \right| #3 \right]}{\mathbb{P}_{#1}\left(\left. #2 \, \right | #3 \right)}}
\newcommand{\CPE}[3][]
{\ifthenelse{\equal{#1}{}}{\PE\left[ \left. #2 \right| #3
    \right]}{\mathbb{E}_{#1} \left[ \left. #2 \right| #3 \right]}}
\newcommandx\cprob[4][1=,2=]{\ensuremath{\PP_{#1}^{#2}\left[ \left. #3 \right|
      #4 \right]}}
\newcommandx\HMCP[2][1=]{
\ifthenelse{\equal{#1}{}}{\PP_{#2}}{\PP_{#1,#2}}
}
\newcommandx\HMCE[2][1=]{
\ifthenelse{\equal{#1}{}}{\PE_{#2}}{\PE_{#1,#2}}
}
\newcommandx\HMCEarg[3][1=]{
\ifthenelse{\equal{#1}{}}{\PE_{#2}\left[#3\right]}{\PE_{#1,#2}\left[#3\right]}
}
\def\loikhi2{\mathbf{\chi^2}}

\newcommandx{\proj}[2]{\ensuremath{\operatorname{proj}\left( \left. #1\right|#2\right)}}

\newcommand{\URoot}{\ensuremath{R}}
\newcommand{\UCov}[1][]%
{%
\ifthenelse{\equal{#1}{}}{\URoot \URoot^t}{\URoot_{#1} \URoot^t_{#1}}%
}

\newcommand{\VRoot}{\ensuremath{S}}
\newcommand{\VCov}[1][]%
{%
\ifthenelse{\equal{#1}{}}{\VRoot \VRoot^t}{\VRoot_{#1} \VRoot^t_{#1}}%
}

\newcommand{\LDX}[2]{\ensuremath{L}}

\newcommand{\postdx}[3][]%
{%
\ifthenelse{\equal{#1}{}}{\ensuremath{\psi_{#2|#3}}}{\ensuremath{\psi_{#1,#2|#3}}}%
}

\newcommand{\epostdx}[3][]%
{%
\ifthenelse{\equal{#1}{}}{\ensuremath{\hat{\psi}_{#2|#3}}}{\ensuremath{\hat{\psi}_{#1,#2|#3}}}%
}

\newcommandx{\predx}[3][1=\bX]{#1_{#2|#3}}   % prediction of boldface X as default
\newcommand{\predpx}[3][]%
{%
\ifthenelse{\equal{#1}{}}{\ensuremath{\varphi_{#2|#3}}}{\ensuremath{\varphi_{#1,#2|#3}}}%
}
\newcommandx\cesp[4][1=,2=]{\ensuremath{{\mathbb E}_{#1}^{#2}\left[ \left. #3 \right| #4 \right]}}

\newcommand{\filt}[2][]%
{%
\ifthenelse{\equal{#1}{}}{\ensuremath{\phi_{#2}}}{\ensuremath{\phi_{#1,#2}}}%
}
\newcommand{\pred}[3][]%
{%
\ifthenelse{\equal{#1}{}}{\ensuremath{\phi_{#2|#3}}}{\ensuremath{\phi_{#1,#2|#3}}}%
}
\newcommand{\post}[3][]%
{%
\ifthenelse{\equal{#1}{}}{\ensuremath{\phi_{#2|#3}}}{\ensuremath{\phi_{#1,#2|#3}}}%
}
\newcommand{\logl}[2][]%
{%
\ifthenelse{\equal{#1}{}}{\ensuremath{\ell_{#2}}}{\ensuremath{\ell_{#1,#2}}}%
}
\newcommand{\lhood}[2][]%
{%
\ifthenelse{\equal{#1}{}}{\ensuremath{\mathrm{L}_{#2}}}{\ensuremath{\mathrm{L}_{#1,#2}}}%
}
\newcommand{\cc}[2][]%
{%
\ifthenelse{\equal{#1}{}}{\ensuremath{c_{#2}}}{\ensuremath{c_{#1,#2}}}%
}
\newcommand{\forvar}[2][]%
{%
\ifthenelse{\equal{#1}{}}{\ensuremath{\alpha_{#2}}}{\ensuremath{\alpha_{#1,#2}}}%
}
\newcommand{\nforvar}[2][]%
{%
\ifthenelse{\equal{#1}{}}{\ensuremath{\bar{\alpha}_{#2}}}{\ensuremath{\bar{\alpha}_{#1,#2}}}%
}

\newcommand{\BK}[2][]%
{%
\ifthenelse{\equal{#1}{}}{\ensuremath{\mathrm{\mathrm{B}}_{#2}}}{\ensuremath{\mathrm{B}_{#1,#2}}}%
}
\newcommand{\filtfunc}[2][]%
{%
\ifthenelse{\equal{#1}{}}{\ensuremath{\tau_{#2}}}{\ensuremath{\tau_{#1,#2}}}%
}
\newcommand{\filtmean}[2][]
{\ifthenelse{\equal{#1}{}}{{\ensuremath{\hat{X}_{#2|#2}}}}{\ensuremath{\hat{X}_{#1,#2|#2}}}
}
\newcommand{\filtcov}[2][]
{\ifthenelse{\equal{#1}{}}{\ensuremath{\Sigma_{#2|#2}}}{\ensuremath{\Sigma_{#1,#2|#2}}}}
\newcommand{\postmean}[3][]
{\ifthenelse{\equal{#1}{}}{\ensuremath{\hat{X}_{#2|#3}}}{\ensuremath{\hat{X}_{#1,#2|#3}}}
}
\newcommand{\postcov}[3][]
{\ifthenelse{\equal{#1}{}}{\ensuremath{\Sigma_{#2|#3}}}{\ensuremath{\Sigma_{#1,#2|#3}}}}

\newcommandx{\QEM}[4][1=,4=]{\ensuremath{\mathcal{Q}_{#1}(#4;#2 \, ; #3)}}

\newcommandx\sequence[3][2=t,3=\zset]{\ensuremath{\left(#1_{#2}\right)_{#2 \in #3 }}}
\newcommandx\dsequence[4][3=t,4=\zset]{\ensuremath{\left( (#1_{#3}, #2_{#3})\right)_{#3 \in #4}}}
\newcommandx{\sequencen}[2][2=n\in\nset]{\ensuremath{\left(#1\right)_{#2}}}

\def\bpm{\left[\begin{matrix}}
\def\epm{\end{matrix}\right]}
\def\bma{\begin{matrix}}
\def\ema{\end{matrix}}

\newcommand{\be}{\begin{equation}}     % begin eqn
\newcommand{\ee}{\end{equation}}        % end eqn

\newcommand{\eg}{\textit{e.g.}}

\newcommandx\lnorm[3][1=]{\left\lVert #2 \right\rVert^{#1}_{#3}}

\newcommandx\supnorm[2][1=]{| #2 |^{#1}_\infty}

\newcommandx\ball[3][1=]{\mathrm{B}_{#1} (#2,#3)}

\newcommandx{\prohosym}[1][1=]{{\boldsymbol\rho}_{#1}}
\newcommandx{\proho}[3][1=]{\prohosym{#1}\left(#2,#3\right)}

\newcommandx{\pp}[1][1=\mu]{\ensuremath{#1\-\mathrm{a.e.}}}

\renewcommand{\-}{\mbox{-}}

\newcommandx{\as}[1][1=\PP]{\ensuremath{#1\-\mathrm{a.s.}}}

\newcommandx{\oscnorm}[3][1=,3=]{\operatorname{osc}^{#1}_{#3}\left(#2\right)}

\newcommandx{\tvdist}[3][1=]{\ensuremath{d^{#1}_{\mathrm{TV}}}(#2,#3)}

\newcommandx{\VnormFunc}[3][1=]{\ensuremath{\left|#2\right|_{\mathrm{#3}}^{#1}}}

\newcommand{\continuousfunctionset}[1]{\mathrm{C}_b(#1)}

\newcommand{\lipschitzfunctionset}[1]{\mathrm{Lip}(#1)}
\newcommand{\boundedlipschitzfunctionset}[1]{\mathrm{Lip}_b(#1)}

\newcommandx\functionsetarg[2][1=]{
\ifthenelse{\equal{#1}{c}}{\continuousfunctionset{\mathsf{#2}}}%fonctions continues
{\ifthenelse{\equal{#1}{bc}}{\mathrm{C}_b(#2)}%fonctions continues born\'{e}es
{\ifthenelse{\equal{#1}{u}}{\mathrm{U}(#2)}%fonctions uniform\'{e}ment continues
{\ifthenelse{\equal{#1}{bu}}{\mathrm{U}_b(#2)}%fonctions uniform\'{e}ment continues born\'{e}es
{\ifthenelse{\equal{#1}{l}}{\lipschitzfunctionset{#2}}%fonctions lipschitz
{\ifthenelse{\equal{#1}{bl}}{\boundedlipschitzfunctionset{#2}}%fonctions lipschitz born\'{e}es
{\mathbb{F}_{#1}(#2)}%le reste
}}}}}}

\newcommandx\functionsetspec[1][1=]{
\ifthenelse{\equal{#1}{c}}{\mathrm{C}}%fonctions continues
{\ifthenelse{\equal{#1}{bc}}{\mathrm{C}_b}%fonctions continues born\'{e}es
{\ifthenelse{\equal{#1}{u}}{\mathrm{U}}%fonctions uniform\'{e}ment continues
{\ifthenelse{\equal{#1}{bu}}{\mathrm{U}_b}%fonctions uniform\'{e}ment continues born\'{e}es
{\ifthenelse{\equal{#1}{l}}{\mathrm{Lip}}%fonctions lipschitz
{\ifthenelse{\equal{#1}{bl}}{\mathrm{Lip}_b}%fonctions lipschitz born\'{e}es
{\mathbb{F}_{#1}}%le reste
}}}}}}

\newcommandx{\taboo}[3][1=,3=]{\left(\leftidx{_#1}{#2}{}\right){^{#3}}}

\newcommandx\vectornorm[2][1=]{\left| #2 \right|^{#1}}  %%%% norme de vecteurs pas fonctions

\newcommand{\ensemble}[2]{\left\{#1\,:\eqsp #2\right\}}
\newcommand{\set}[2]{\ensemble{#1}{#2}}

\newcommandx{\plim}[1]{\ensuremath{\stackrel{#1\-\text{prob}}{\longrightarrow}}}
\newcommandx{\dlim}[1]{\ensuremath{\stackrel{#1}{\Longrightarrow}}}

\newcommandx\measureset[3][1=\mathrm{s},3=]{\mathbb{M}^{#3}_{#1}(#2)}
\newcommandx\measuresetmetric[2][1=1]{\mathbb{M}_{#1}(\mathcal{B}(\mathsf{#2}))}  %%% espace de mesures sur un metric muni de sa tribu borelienne
\newcommandx\measuresetspec[1][1=\mathrm{s}]{\mathbb{M}_{#1}}

\newcommand{\abs}[1]{\left\vert #1 \right\vert}

\newcommandx\canonicalkernel[1][1=P]{\mathbb{K}_{#1}}
\def\bi{\mathbf{i}}

\usepackage[textwidth=2cm,textsize=footnotesize]{todonotes}    %  [disable] shuts it off

\begin{document}

\title{New penalized criteria for smooth non-negative tensor factorization with missing entries}

\author{\IEEEauthorblockN{
    Amaury Durand\IEEEauthorrefmark{1}\IEEEauthorrefmark{2},
    Fran\c{c}ois Roueff\IEEEauthorrefmark{1},
    Jean-Marc Jicquel\IEEEauthorrefmark{2},
    and Nicolas Paul\IEEEauthorrefmark{3}} \\~\\
  \IEEEauthorblockA{\IEEEauthorrefmark{1} LTCI, Telecom Paris, Institut Polytechnique de Paris. 19 Place Marguerite Perey, 91120 Palaiseau, France.} \\
  \IEEEauthorblockA{\IEEEauthorrefmark{2} EDF R\&D, TREE, E36, Lab Les Renardieres. Ecuelles, 77818 Moret sur Loing, France.} \\
  \IEEEauthorblockA{\IEEEauthorrefmark{3}EDF R\&D, PRISME.  6 quai Watier, 78400 Chatou, France. }
\thanks{This work has been submitted to the IEEE for possible publication. Copyright may be transferred without notice, after which this version may no longer be accessible.}
}



\maketitle

\begin{abstract}
  Tensor factorization models are widely used in many applied fields
  such as chemometrics, psychometrics, computer vision or
  communication networks. Real life data collection is often subject to
  errors, resulting in missing data. Here we focus in understanding
  how this issue should be dealt with for non-negative tensor
  factorization. We investigate several criteria used for non-negative
  tensor factorization in the case where some entries are missing. In
  particular we show how smoothness penalties can compensate the
  presence of missing values in order to ensure the existence of an
  optimum. This lead us to propose new criteria with efficient
  numerical optimization algorithms. Numerical experiments are
  conducted to support our claims.
\end{abstract}
 
\begin{IEEEkeywords}
Non-negative tensor decomposition, missing values, Tensor completion, smoothness, PARAFAC, CP decomposition.
\end{IEEEkeywords}

\section{Introduction}
The ever growing literature on multi-way data has shown the
effectiveness of tensor factorization models in many domains ranging
from psychometrics and chemometrics to signal processing and machine
learning (see \cite{KoBa09,Cichocki-NTF,Sidiropoulos17-tensor-sp-ml}
and the references therein). Key strengths of these models include 1)
their flexibility in taking into account prior knowledge on the data
such as non-negativity, sparsity or smoothness and 2) their ability to
cope with missing data.  As a result, they have proven very useful in
real world applications, where factors are related to practical
features.  On the other hand, the theoretical study of tensor
factorization is still a challenging research topic. In particular,
finding the true factorization of a tensor is generally NP-Hard
\cite{Hillar-tensor-NP-Hard} and finding an approximated factorization
requires solving a non-convex optimization problem often dealt with
using an iterative algorithm based on alternating minimization or on
gradient descent. Since the problem is non-convex, few guarantees
exist on the convergence of the optimization method used except that
the objective function decreases at each iteration and that the
algorithm converges to a local minimum. A question which naturally
arises in this context is the existence of a global
minimum. Unfortunately, this existence is not always guaranteed and
discrepancies have been explored in both theoretical and experimental
works especially for the popular CANDECOMP/PARAFAC (CP) decomposition
(see e.g. \cite{Silva08-illposed-tensor,Lim09-NTF} and the references
therein). However, it is common knowledge in the inverse problem and
statistical learning communities, that ill-posed problems can be
handled by controlling the complexity of the solution space.  This can
be done, for example, using regularization
\cite{Tikhonov-ill-posed,Vapnik-statistical-learning-theory}. In the
context of tensor factorization, is has been shown in
\cite{lim2005optimal, Lim09-NTF} that adding non-negativity
constraints ensures existence of a global solution for the CP
decomposition. This decomposition is known as non-negative tensor
factorization (NTF). In this paper, we provide similar guarantees for
the NTF problem in the case where some, possibly many, entries are
missing. This case is of particular interest because missing data are
very common in practical settings where the data collection can be
subject to errors. Another interest of dealing with missing entries is
that it allows to use cross validation methods to select
hyperparameters such as penalty parameters or the tensor rank.

\subsection{Related work}
The literature on tensor factorization from the past decades has given
rise to a wide range of algorithms among which some are adapted to
missing values and/or additional constraints. Methods handling missing
values usually fall into one of the three following categories :
imputation, weighted least squares and probabilistic models. In the
first case, missing entries are estimated at each iteration resulting
in an EM-like algorithm
\cite{Bro97-parafac,ANDERSSON-Nway-toolbox}. In the second case the
squared error is weighted with binary weights representing missing and
observed entries (see
e.g. \cite{TOMASI-parafac-missing-values,ACAR-tf-missing-data}). In
the last case, prior distributions are proposed for the factors and
their parameters are estimated from the observed data
\cite{Xiong10-BPTF,pmlr-v32-rai14,Zhao16-bayesian-tensor}. The problem
of missing entries is also closely related to tensor completion where,
in addition to the factors, one usually also tries to estimate a full
tensor which coincides with the data tensor on observed entries (see
e.g \cite{Song19-tensorcompletion} for a recent survey). Because of
the need for efficient algorithms to deal with large amount of data,
the literature on tensor factorization is dominated by algorithmic
considerations, especially when missing values are taken into
account. 

The effect of using smoothness constraints or penalties has also been
thoroughly explored both for tensor factorization and tensor
completion \cite{TIMMERMAN2002447, Reis02-parafac-spline,
  YOKOTA15-smoothCP, Yokota16-tensor-completion-smooth,
  Li17-low-rank-completion,Imaizumi17-tensor-smoothness,Sadowski18-imagecompletion,durand21-SmoothNTF}. Such
approaches are generally used for numerical purposes to regularize
ill-posed optimization problems and for statistical purposes to
compensate overfitting by incorporating prior knowledge or assumption
on the model. This can be beneficial for the interpretation of the
factors \cite{TIMMERMAN2002447, Reis02-parafac-spline, Henriet19NILM,
  durand21-SmoothNTF} but also for the accuracy of the factorization
or completion. For example, total variation constraints are widely
used to deal with natural images as they are able to capture their
smoothness structure
\cite{gousseau01,Yokota16-smoothCP-tensor-completion,Li17-low-rank-completion}. There
are two main strategies to impose smoothness on the factors. The first
strategy, which is studied in this paper, consists in adding a penalty
term in the loss. Usual penalties for smoothness involve the total
variation norm or the $L^2$ norm of the second derivative for spline
smoothing, \cite{Reis02-parafac-spline,
  Yokota16-smoothCP-tensor-completion,Li17-low-rank-completion,
  Henriet19NILM, durand21-SmoothNTF}. The second strategy consists in
representing the factors within specific lower dimensional spaces,
using splines, polynomials or kernels \cite{TIMMERMAN2002447,
  Reis02-parafac-spline, Zdunek14-splineNMF, YOKOTA15-smoothCP,
  Yokota16-tensor-completion-smooth, Amini17-functional-cp,
  Imaizumi17-tensor-smoothness,Sadowski18-imagecompletion,
  HAUTECOEUR2020256}.

\subsection{Notation} 
The interval of integers between $i$ and $j$ is denoted by
$\iseg{i,j} := \{i, \cdots, j\}$.  We use bold capital letters to
denote tensors and matrices and bold lowercase letters for vectors.
Standard font is used for the entries of the tensors, matrices and
vectors. For example,
$\bA = [\ba_1, \cdots, \ba_R] \in \rset^{I \times R}$ means that
$\ba_r \in \rset^I$ is the $r$-th column of $\bA$ and the $(i,r)$-th
entry of $\bA$ is denoted by $A_{i,r}$ or $a_{i,r}$. A tensor
$\bX \in \rset^{I_1 \times \cdots \times I_N}$ is indexed by a vector
of integers
$\bi = i_{1:n}:=(i_1, \cdots, i_N) \in \prod_{n=1}^N \iseg{1,I_n}$ and
we write $X_\bi = X_{i_1, \cdots, i_n}$. We denote the outer product
between vectors by $\circ$ and the Hadamard product between tensors by
$\oast$. We refer to \cite{KoBa09} for the definitions of these usual
tensor operations.  We recall that the Frobenius scalar product of two
tensors $\bX, \bY \in \rset^{I_1\times\cdots\times I_N}$ is defined as
$\pscal{\bX}{\bY}_F = \sum_{\bi} X_{\bi} Y_{\bi}$ and we denote by
$\norm{\cdot}_F$ its induced norm. The total variation $p$-norm of a
vector $\ba \in \rset^{I}$ is
$\norm{\ba}_{{\rm TV},p} = \left(\sum_{i=1}^{I-1} \abs{a_i -
    a_{i+1}}^p\right)^{1/p}$ for $p \in [1,\infty)$ and
$\norm{\ba}_{{\rm TV}, \infty} = \max_{1 \leq i \leq I-1} \abs{a_i -
  a_{i+1}}$.  For a given norm or semi-norm $\nu$ on $\rset^{I}$ we
respectively denote
$\mathbb{S}_{\nu}^+ := \set{\ba \in \rset_+^{I}}{\nu(\ba) = 1}$ and
$\mathbb{B}^+_{\nu} := \set{\ba \in \rset_+^{I}}{\nu(\ba)\leq 1}$ the
unit positive sphere and the unit positive ball. We also write for any spaces
$\cX_1,\cdots,\cX_N$,
$x^{(1:N)} := (x^{(1)},\cdots,x^{(N)}) \in \prod_{n=1}^N \cX_n$.
Finally, throughout this paper we will denote the $N$ dimensional grid
of indices by $\cI = \prod_{n=1}^N \iseg{1,I_n}$ and, for any tensor
$\bX \in \rset^{I_1,\cdots,I_N}$ and any set $A \subset \rset$, we
denote the set %$\set{\bi\in\cI}{X_{\bi} \in A}$
of indices at which $\bX$'s entries fall into $A$ by $\{\bX \in A\}$.
Specifically, we denote $\{\bX = x\} = \{\bX \in \{x\}\}$ the set of
indices with entries equal to $x$ or
$\{\bX > x\} = \{\bX \in \osego{x,+\infty}\}$.

The remaining of this paper is organized as follows. In
\Cref{sec:losses}, we describe the tensor factorization problem and
introduce two losses and their related optimization algorithms. The
main theoretical contributions are gathered in
\Cref{sec:global-optimum}. Numerical experiments are conducted in
\Cref{sec:expe-swntf} for comparing the optimization algorithms of
\Cref{sec:losses}. Proofs are postponed to \Cref{sec:proofs}.

\section{Losses and optimization algorithms}\label{sec:losses}
A tensor $\bY \in \rset^{I_1\times \cdots \times I_N}$ is said to admit
a non-negative tensor factorization with rank $R$ if there exists a
sequence of $N$ factor matrices
$\bA^{(1:N)} \in \prod_{n=1}^N\rset_+^{I_1 \times R}$ such that
\begin{equation}\label{eq:ntf-exact}
\bY = \sum_{r=1}^R \ba_r^{(1)} \circ \cdots \circ \ba_r^{(N)} \; ,
\end{equation}
where we recall that $\ba_r^{(n)}$ is the $r$-th column of
$\bA^{(n)}$. In order to avoid scaling indeterminacy in the tensor
factorization, it is common to use an equivalent formulation of
\eqref{eq:ntf-exact} using normalized factors based on $N$ norms
$\nu_1, \cdots, \nu_N$ on $\rset^{I_1},\cdots,\rset^{I_N}$. Let us
set, for all $n\in\iseg{1,N}$ and $r\in\iseg{1,R}$,
$\lambda_r = \prod_{n=1}^n \nu_n(\ba_r^{(n)})$ and
$\tilde{\ba}_r^{(n)} = \frac{\ba_r^{(n)}}{\nu_n(\ba_r^{(n)})}$ if
$\nu_n(\ba_r^{(n)}) > 0$ and any arbitrary element of
$\mathbb{S}_{\nu_n}^+$ otherwise. Then, the factorization
\eqref{eq:ntf-exact} can be written as
\begin{equation}\label{eq:ntf-exact-normalized}
\bY = \sum_{r=1}^R \lambda_r \tilde{\ba}_r^{(1)} \circ \cdots \circ \tilde{\ba}_r^{(N)} \; .
\end{equation}
We study the problem of approximating a tensor
$\bX \in \rset^{I_1\times\cdots\times I_N}$ by a finite rank tensor
using parameterization \eqref{eq:ntf-exact} or
\eqref{eq:ntf-exact-normalized} in the case where some entries of the
tensor are missing. Let $\bW \in \rset_+^{I_1, \cdots, I_N}$ be a
tensor of weights and define
$$
L_\bW(\bA) :=  \norm{\bW \oast \left(\bX - \sum_{r=1}^R \ba_r^{(1)} \circ \cdots \circ \ba_r^{(N)}\right)}_F^2 \; ,
$$
for all $\bA=\bA^{(1:N)} \in \prod_{n=1}^N \rset_+^{I_n \times R}$.
We also define the normalized equivalent, 
$$
\tilde{L}_{\bW}(\blambda,\tilde{\bA}) := \norm{\bW \oast \left(\bX - \sum_{r=1}^R \lambda_r \tilde{\ba}_r^{(1)} \circ \cdots \circ \tilde{\ba}_r^{(N)}\right)}_F^2 \; ,
$$
for all $\blambda\in \rset_+^R$ and $\tilde{\bA}=\tilde{\bA}^{(1:N)}\in \prod_{n=1}^N \left(\mathbb{S}_{\nu_n}^+\right)^R$, where we view $\left(\mathbb{S}_{\nu_n}^+\right)^R$ as the set of matrices in $\rset_+^{I_n\times R}$ whose columns are valued in $\mathbb{S}_{\nu_n}^+$. 
In this framework, it is implicitly assumed that the zero entries of $\bW$ indicate missing entries and $\bW$ is usually taken as a binary tensor.

We also introduce two penalties based on $N$ semi-norms
$\mu_1, \cdots, \mu_N$ defined on $\rset^{I_1}, \cdots, \rset^{I_N}$
respectively. Namely, given two integers $d,p$ with $d \geq p \geq 1$
and $\balpha \in \rset_+^N$, we define, for all $\blambda\in \rset_+^R$ and $\tilde{\bA}=\tilde{\bA}^{(1:N)}\in \prod_{n=1}^N \left(\mathbb{S}_{\nu_n}^+\right)^R$,
  \begin{equation}\label{eq:frscp}
  \tilde{\cP}_{\balpha}(\blambda, \tilde\bA) := \sum_{n=1}^N \alpha_n \sum_{r=1}^R \lambda_r^d \, \mu_n^p(\tilde\ba_r^{(n)}) \; ,
\end{equation}
where we use the notation $\mu_n^p(\ba) =
\left(\mu_n(\ba)\right)^p$. This includes the penalty of
\cite{Yokota16-smoothCP-tensor-completion} by taking $d=2$ and
$\mu_n = \norm{\cdot}_{{\rm TV},p}$ for $p = 1,2$. We propose a new
penalty, defined as the unnormalized equivalent of \eqref{eq:frscp}, which reads as
\begin{equation}\label{eq:frscp-grad}
  \cP_{\balpha}(\bA) := \sum_{n=1}^N \alpha_n \sum_{r=1}^R
 \nu_n^{d-p}(\ba_r^{(n)}) \, \mu_n^p(\ba_r^{(n)})\,\prod_{m\neq n}\nu_m^d(\ba_r^{(m)})\;,
\end{equation}
for all $\bA=\bA^{(1:N)} \in \prod_{n=1}^N \rset_+^{I_n\times R}$.

Then, let $f_{\bW,\balpha} := L_\bW + \cP_{\balpha}$  and $\tilde{f}_{\bW,\balpha} := \tilde{L}_\bW+ \tilde{\cP}_{\balpha}$ and consider the two  following equivalent optimization problems
\begin{equation}\label{eq:wntf-pen}
  \begin{split}
    & \min_{\bA^{(1:N)}} f_{\bW,\balpha}(\bA^{(1:N)}) \\
    &\text{s.t. }  \forall n \in \iseg{1,N}, \, \bA^{(n)} \in \rset_+^{I_n\times R} \; ,
  \end{split}
\end{equation}
and 
\begin{equation}\label{eq:wntf-pen-normalized}
  \begin{split}
    & \min_{\blambda, \tilde\bA^{(1:N)}} \tilde{f}_{\bW,\balpha}(\blambda,\tilde{\bA}^{(1:N)})  \\
    &\text{s.t. } \forall r \in \iseg{1,R},
    \lambda_r \geq 0  \text{ and } \forall n \in \iseg{1,N}, \tilde\ba_r^{(n)} \in \mathbb{S}_{\nu_n}^+
    \; . 
  \end{split}
\end{equation}

The fact that the $\tilde{\ba}_r^{(n)}$'s are constrained individually
in Problem~\eqref{eq:wntf-pen-normalized} naturally leads to the
Hierarchical Alternating Least Squares (HALS) optimization method (see
\eg\ \cite{Cichocki-NTF,Yokota16-smoothCP-tensor-completion}). This
method consists in minimizing $\tilde{f}_{\bW,\balpha}$ alternatively
in the $\tilde{\ba}_{r}^{(n)}$'s and its updates are recalled in
Algorithm~\ref{alg:HALS} where we have defined
$\bX^{(r)} := \bX - \sum_{s \neq r} \lambda_s \tilde{\ba}_s^{(1)}
\circ \cdots \circ \tilde{\ba}_s^{(N)}$ and
$\tilde{f}_{\bW,\balpha}^{(r,n)}(\tilde{\ba}) := \norm{\bW \oast
  (\bX^{(r)} - \lambda_r \tilde{\ba}_r^{(1)} \circ \cdots
  \tilde{\ba}_r^{(n-1)} \circ \tilde{\ba} \circ \tilde{\ba}_r^{(n+1)}
  \circ \cdots \circ \tilde{\ba}_r^{(N)})}_F^2$
$+ \alpha_n \lambda_r^d \mu_n^p(\tilde{\ba})$. 
\begin{algorithm}
  \KwData{$\bX$, $\bW$ and initial values for $\blambda, \tilde{\bA}^{(1)}, \cdots, \tilde{\bA}^{(N)}$}
  $\bE = \bX - \sum_{r=1}^R \lambda_r \tilde{\ba}_r^{(1)} \circ \cdots \circ \tilde{\ba}_r^{(N)}$ \\
  \Repeat{Change of value of $\tilde{f}_{\bW,\alpha}$ is sufficiently small}{
    \For{$r=1, \cdots, R$}{
      $\bX^{(r)} = \bE + \lambda_r \tilde{\ba}_r^{(1)} \circ \cdots \circ \tilde{\ba}_r^{(N)}$ \\
      \For{$n=1, \cdots, N$}{
        $\displaystyle \tilde{\ba}_r^{(n)} = \argmin_{\tilde{\ba} \in \mathbb{S}_{\nu_n}^+} \tilde{f}_{\bW,\balpha}^{(r,n)}(\tilde{\ba})$
      }
      $\displaystyle \lambda_r = \argmin_{\lambda \geq 0} \norm{\bW \oast (\bX^{(r)} - \lambda \tilde{\ba}_r^{(1)} \circ \cdots \circ \tilde{\ba}_r^{(N)})}_F^2 $\\
      $\bE = \bX^{(r)} - \lambda_r \tilde{\ba}_r^{(1)} \circ \cdots \circ \tilde{\ba}_r^{(N)}$ 
    }
  }
  \Return{$\blambda, \tilde{\bA}^{(1)}, \cdots, \tilde{\bA}^{(N)}$}
 \caption{HALS algorithm \label{alg:HALS} }
\end{algorithm}

The update for $\tilde{\ba}_r^{(n)}$ in Algorithm~\ref{alg:HALS}
requires minimizing the function $\tilde{f}_{\bW,\balpha}^{(r,n)}$ on
the non-negative sphere $\mathbb{S}_{\nu_n}^+$. In the case where
there is no non-negativity constraint,
\cite{Yokota16-smoothCP-tensor-completion} proposes to use a projected
gradient method. In our case, we propose to first solve
$\min_{\tilde{\ba} \in \rset_+^{I_n}}
\tilde{f}_{\bW,\balpha}^{(r,n)}(\tilde{\ba})$ using a gradient-based
method with bound constraints such as L-BFGS-B and then normalize the
result by its $\nu_n$-norm. The update in $\lambda_r$ writes as
$\lambda_r = \left[\frac{\pscal{\bW \oast \bX^{(r)}}{\bW \oast
      (\tilde\ba_r^{(1)} \circ \cdots \circ
      \tilde\ba_r^{(N)})}_F}{\norm{\bW \oast (\tilde\ba_r^{(1)} \circ
      \cdots \circ \tilde\ba_r^{(N)})}_F^2}\right]_+$.

On the other hand, the constraints in Problem~\eqref{eq:wntf-pen}
reduce to a bound constraint on a vectorization of $\bA$. Hence, in
the case where the gradients of $\mu_n$ and $\nu_n$ are available, we
can follow the approach of \cite{ACAR-tf-missing-data} and use a
gradient-based method with bound constraints such as L-BFGS-B.

In the next section, we study the existence of a global optimum for
Problem~\eqref{eq:wntf-pen-normalized}. The existence of a global
optimum for Problem~\eqref{eq:wntf-pen} follows immediately as they
are two different parameterizations of the same criterion. However,
using the normalized version allows for stronger results, such as coercivity,
which are not possible in the unnormalized version as already noted in \cite{Lim09-NTF}.

\section{Existence of a global optimum for the weighted NTF}\label{sec:global-optimum}
\label{sec:theoretical-results}
In this section, we consider the normalized version of the weighed NTF
problem, i.e. Problem~\eqref{eq:wntf-pen-normalized}, where, for all
$n\in\iseg{1,N}$, $\nu_n$ is a norm on $\rset^{I_n}$ and $\mu_n$ is a
semi-norm on $\rset^{I_n}$. To simplify the presentation, we set
$\Theta := \rset_+^R \times \prod_{n=1}^N
\left(\mathbb{S}_{\nu_n}^+\right)^R$ and
$\btheta = (\blambda, \bA^{(1:N)}) \in \Theta$.  We investigate the
existence of a minimum with minimal assumptions on the semi-norms used
in the penalty term. Namely, for $n=1,\dots,N$, we define
\begin{align}
  \label{eq:inf-mu-zero-def}
& \psi_n :=  \sup\set{\nu_n(\ba)}{\ba \in \mathbb{B}^+_{\mu_n}} \;,\\
  \label{eq:continuity-norms-def}
& \psi_n^{\prime} := \sup\set{\nu_n(\ba)}{\ba \in
                                    \mathbb{B}^+_{\mu_n},\exists i\in\iseg{1,I_n}\,,a_i=0}\;,
\end{align}
and consider the following admissibility condition for the semi-norm
$\mu_n$.
\begin{enumerate}[label=(\textbf{AC})]
\item\label{item:assumption-seminorms} We have
  $\psi_n^{\prime}<\psi_n=\infty$,
\end{enumerate} 
Condition $\psi_n=\infty$ simply says that $\mu_n$ is
not a norm on the positive cone (otherwise it would be equivalent to the norm $\nu_n$). The
condition $\psi_n^{\prime}<\infty$ says that $\mu_n$ behaves as a norm
on the positive sub-cone that have at least one zero entry.
Condition~\ref{item:assumption-seminorms} holds for semi-norms typically used in smoothness
penalties, as shown by the following lemma.
\begin{lemma}\label{lem:smoothness-continuity-norms}
 Condition~\ref{item:assumption-seminorms} holds for $\mu_n$ defined as one of the following semi-norms.
  \begin{enumerate}[label=(\roman*)]
  \item\label{itm:tv-p} For all $\ba \in \rset^{I_n}$, $\mu_n(\ba) = \norm{\ba}_{{\rm TV},p}$ for some  $p\in [1,+\infty]$.
  \item\label{itm:spline} For all $\ba \in \rset^{I_n}$, $\mu_n(\ba) = \left(\int_0^1 (a''(u))^2 \, \rmd u\right)^{1/2}$, where $a : [0,1] \to \rset$ is the natural cubic spline such that, for all $i\in\iseg{1,I_n}$, $a(u_i) = a_i$ for some $0 < u_1 \leq \cdots \leq u_{I_n} < 1$.  
  \end{enumerate}
\end{lemma}
In the next theorem, we provide a necessary
and sufficient condition for $\tilde{f}_{\bW,\balpha}$ to be
\emph{coercive} on $\Theta$, which means that
$\tilde{f}_{\bW,\balpha}(\theta)$ tends to $+\infty$ as the norm of
$\theta$ goes to $+\infty$.

The necessary and sufficient condition relies on the
following definition of cylinders.
\begin{definition}
  For any non-empty subset $\cN \subset \iseg{1,N}$, a $\cN$-cylinder
  is a set defined as 
  \begin{equation}\label{eq:cylinder}
  \cC((j_n)_{n\in\cN}) := \set{\bi\in\cI}{i_n=j_n, \; \forall n \in\cN} \; ,
\end{equation}
for some $(j_n)_{n\in\cN} \in \prod_{n\in\cN}\iseg{1,I_n}$. In
particular, the whole set $\cI$ the unique $\emptyset$-cylinder.
\end{definition}
Then the following result holds.
\begin{theorem}\label{thm:coercivity-penalized}
  Let $\balpha \in \rset_+^N$ and suppose that, for all $n\in\{\balpha > 0\}$, $\mu_n$ satisfies~\ref{item:assumption-seminorms}. 
  Then the two following assertions are equivalent.  
  \begin{enumerate}[label=(\roman*)]
  \item\label{itm:no-missing-cylinder} The set $\{\bW = 0\}$ contains no $\{\balpha = 0\}$-cylinder.
  \item\label{itm:frscp-compact-level-set} The function
    $\tilde{f}_{\bW,\balpha}$ is coercive on $\Theta$. 
  \end{enumerate}
  In this case, both $\tilde{f}_{\bW,\balpha}$ and $f_{\bW,\balpha}$
  admit global minima on
  $\Theta$ and on
  $\prod_{n=1}^N \rset_+^{I_n \times R}$, respectively.
\end{theorem}
We assume that $\mu_n$ satisfy \ref{item:assumption-seminorms} only
for $n\in\{\balpha>0\}$ because, for $n\in\{\balpha=0\}$, $\mu_n$
vanishes in the penalty~(\ref{eq:frscp}).  Also note that
Assertion~\ref{itm:no-missing-cylinder} means that no
$\{\balpha = 0\}$-cylinder is missing and is not very restrictive. For
example, if all modes are penalized, i.e.
$\{\balpha = 0\} = \emptyset$, then
Assertion~\ref{itm:no-missing-cylinder} does not hold if and only if
all entries are missing. In the experiments, we study the case of
color image completion where $N=3$ and the first two modes correspond
to pixels and the third corresponds to the color channel. In this
case, we penalize only the first two modes, i.e.
$\{\balpha = 0\} = \{3\}$. This means that
Assertion~\ref{itm:no-missing-cylinder} does not hold if and only
if an entire color channel is missing.

\section{Experimental results}\label{sec:expe-swntf}
In this section, we compare the two optimization problems
\eqref{eq:wntf-pen} and \eqref{eq:wntf-pen-normalized} with the
penalties defined in \eqref{eq:frscp} and \eqref{eq:frscp-grad}
respectively where we take $p=d=2$. The first optimization problem is
solved using the \texttt{scipy} implementation of L-BFSG-B. The second
optimization problem is solved using Algorithm~\ref{alg:HALS}. We
propose two experiments. In the first one, we try to recover factors
from an incomplete and noisy observation of a tensor of the type
\eqref{eq:ntf-exact-normalized} with $N=3$. In the second one, we
apply the algorithms for color image completion.  In both experiments
and each of the algorithm, we take $\balpha = [\alpha,\alpha,0]^\top$
where $\alpha > 0$, and $\nu_n = \norm{\cdot}_2$ for all $n$. The
maximum number of iterations is set to $10^4$ and the iterations are
stopped if the relative improvement of the loss is lower than
$10^{-6}$. One iteration of the HALS algorithm consists of the steps
in the {\bf repeat}--{\bf until} box of Algorithm~\ref{alg:HALS} and
one iteration of the L-BFSG-B method consists in one update of the
gradient descent. In this case, during one iteration, all the factors
are updated. The computational performances algorithms are compared
using the average computing time per iteration denoted by TPI. All the
experiments are run on a Linux Workstation with 40 Intel Xeon E5-2630
v4 2.20 GHz processors.

\subsection{Factor estimation on toy data}
In this experiment, we construct a tensor
$\bX = \bY + \sigma \bE \in \rset^{I \times I \times I}$ where $\bY$
is as in \eqref{eq:ntf-exact-normalized} with $I = 50$, $N=3$ and
$R = 5$ and $\bE$ has a standard normal entries. The factors
$\ba_r^{(3)}$ are sampled uniformly on $[0,1]$ and the factors
$\ba_r^{(1)}$ and $\ba_r^{(2)}$ are constructed by taking non-negative
random linear combination of $7$ B-Spline functions of order $4$. With
this construction, the factors are non-negative and smooth for $n=1,2$
and we allow the $\ba_r^{(2)}$'s to vanish on some intervals (see the
first column of \Cref{fig:toy}). The standard deviation $\sigma > 0$
is computed as in \cite{TOMASI-parafac-missing-values,
  ACAR-tf-missing-data}, i.e.
$\sigma = (100 / \nu - 1)^{-1/2} \frac{\norm{\bY}_F}{\norm{\bE}_F}$,
and we take $\nu = 10$. Since the true factors are generated from
B-splines, we take $\mu_n$ as in \ref{itm:spline} of
\Cref{lem:smoothness-continuity-norms}. We generated $25\%, 50\%$ and
$70\%$ percent of missing data which are drawn randomly and uniformly
on the grid $\iseg{1,I}^3$. For the unnormalized problem, we use an
SVD-based initialization and use its normalized equivalent for the
normalized problem. To evaluate the output
$\hat{\bY} = \sum_{r=1}^R \hat\lambda_r \hat\ba_r^{(1)} \circ
\hat\ba_r^{(2)} \circ \hat\ba_r^{(3)}$, we use the normalized mean
square error, NMSE
$=\frac{\norm{\bY - \hat{\bY}}_F^2}{\norm{\bY}_F^2}$ and the
similarity score, SIM
$=\max_{\sigma \in \mathscr{S}_R} \frac{1}{R}\sum_{r=1}^R
\prod_{n=1}^N \pscal{\tilde{\ba}_r^{(n)}}{\hat\ba_{\sigma(r)}^{(n)}}$,
where $\mathscr{S}_R$ denotes the set of permutations of
$\iseg{1,R}$. Note that, since the $\tilde{\ba}_r^{(n)}$'s and
$\hat{\ba}_r^{(n)}$'s have unit $2$-norm, we have SIM $\in [0,1]$. A
lower NSME is interpreted as a better prediction of $\Y$'s entries
while a higher SIM is interpreted as a better estimation of its
factors. An example of reconstruction is shown in \Cref{fig:toy} where
we observe that the gradient method is able to reconstruct the factors
even when the proportion of missing data is high. In the remaining of
this section, we discuss the choice of $\alpha$ and compare the
methods with various values of $I$.

\subsubsection{Data-driven selection of $\alpha$}
We propose to evaluate the selection of $\alpha$ using a
$5$-fold cross validation criterion (CV), which amounts to
arbitrarily introduce additional missing values in the objective
function and evaluate their prediction errors. More precisely, we
generate randomly $5$ binary masks $(\bW_k)_{k\in \iseg{1,5}}$ such
that the sets $\left(\{\bW_k > 0\}\right)_{k\in\iseg{1,5}}$ create a
partition of $\cI$. For each fold $k\in\iseg{1,5}$ and each value of
$\alpha \in \{0, 10^{-4}, 10^{-3}, 10^{-2}, 10^{-1}, 1, 10\}$, we
minimize $f_{\bW \oast \bW_k,\balpha}$ (or
$\tilde{f}_{\bW \oast \bW_k,\balpha}$) and the cross validation score
is computed using $f_{\bW \oast ({\bf 1} - \bW_k),\balpha}$ (or
$\tilde{f}_{\bW \oast ({\bf 1} - \bW_k),\balpha}$), where ${\bf 1}$ is
the tensor with all entries equal to $1$. We compare the cross
validation scores with an oracle selection of $\alpha$ given by the
NMSE and SIM scores which use the ground truth. The results are
gathered in \Cref{fig:alpha} where we observe that the selection with
the cross validation score seems consistent with the oracle based on
NMSE, which indicates that CV is a suitable parameter selection method
in our context. Note that the oracle selection using the SIM score can
give different optimal values for $\alpha$. This can be explained by
the fact that the SIM score is more sensitive to small differences
between the estimated factors and the true factors.

\begin{figure}[h]
  \centering
  \includegraphics[width=0.5\textwidth]{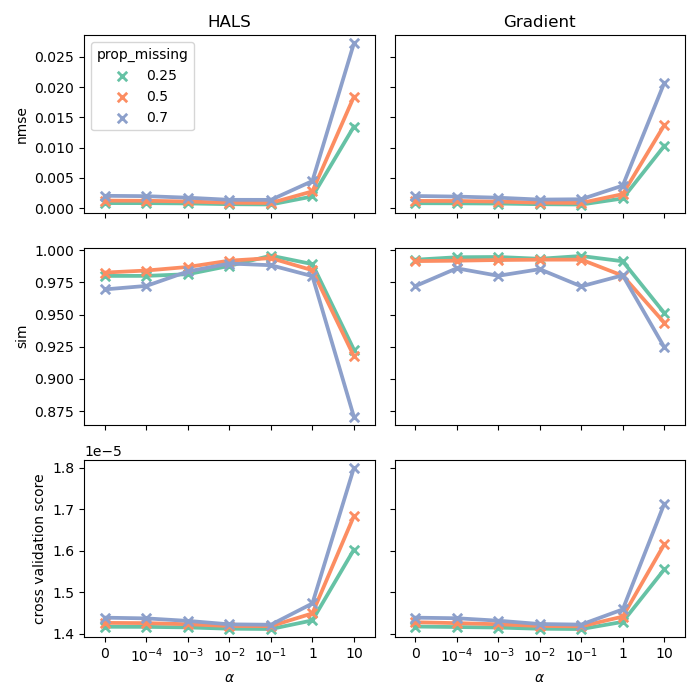}
  \caption{CV selection of $\alpha$ (bottom), compared to the SIM (middle) and
    NMSE (top) critera. \label{fig:alpha}}
\end{figure}
\subsubsection{Comparison with various dimensions}
We now compare the proposed methods for various values of $I$. We use
the same ground truth represented by the first column of
\Cref{fig:toy} where we interpolated the factors to achieve higher
values of $I$. The comparison is made in a best case scenario where,
for each value of $I$, each proportion of missing data and each model,
we use the oracle selection of
$\alpha \in \{0, 10^{-4}, 10^{-3}, 10^{-2}, 10^{-1}, 1, 10\}$ based on
the NMSE score. The NMSE and SIM scores and the TPI are gathered in
\Cref{fig:scores}. The two methods give very similar results for the
NMSE and SIM whose values indicate almost perfect reconstruction in
all cases. We observe that the reconstruction tends to be better for
large values of $I$ which is expected since the difficulty of the
problem decreases as $I$ increases (see
\cite{ACAR-tf-missing-data}). The advantage of the gradient-based
method is, however, highlighted by the TPI, especially when the
dimension increases. Note that we observed that the value of $\alpha$
does not affect much the TPI and therefore the comparison displayed in
\Cref{fig:scores} is representative of any value of $\alpha$.
\begin{figure}[h]
  \centering
  \includegraphics[width=0.5\textwidth]{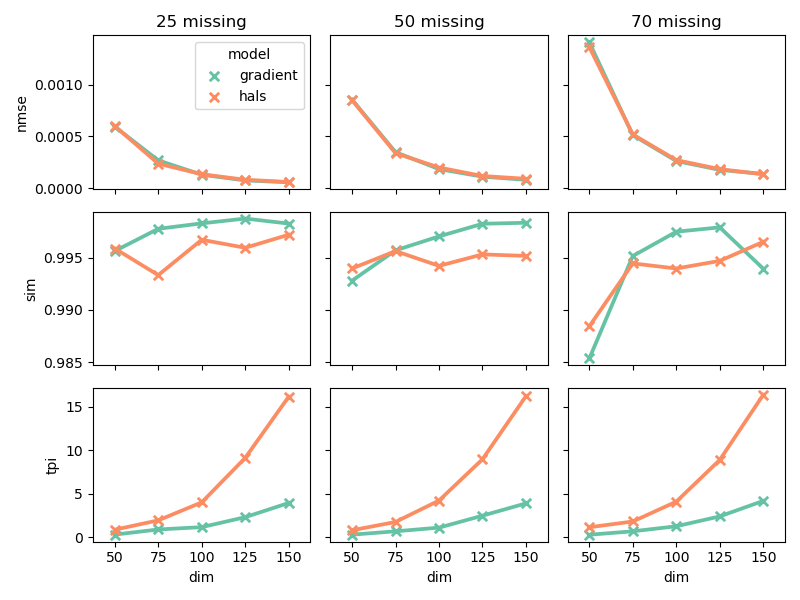}
  \caption{Best case scenario for different dimensions. \label{fig:scores}}
\end{figure}

\begin{figure*}[h]
  \centering
  \includegraphics[width=\textwidth]{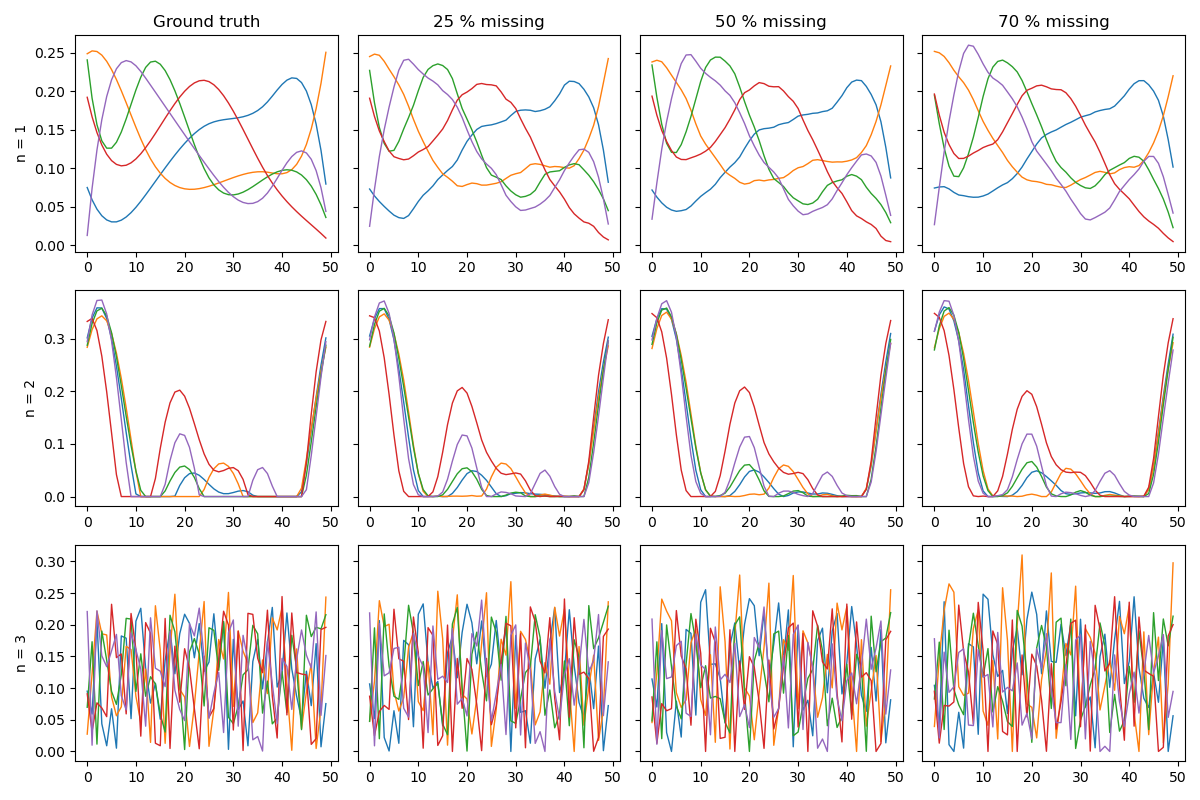}
  \caption{Estimated factors with gradient-based method and $\alpha = 0.1$.  Rows represent the modes (i.e. $n$) and colors represent the components (i.e. $r$). \label{fig:toy}}
\end{figure*}

\subsection{Color image completion}
In this experiment, we apply both methods to $3$ color images
(barbara, baboon and giant) of size $256 \times 256 \times 3$. Missing
data are generated in two ways. In the first case, we remove all the
color channels for $80\%$ of the pixels selected randomly and
uniformly.  In the last case, we remove all the color channels for a
mask obtained by scribbling the image. Note that, in all of these
cases, Assertion~\ref{itm:no-missing-cylinder} of
\Cref{thm:coercivity-penalized} holds. We use the quadratic variation
penalty as in \cite{Yokota16-smoothCP-tensor-completion}, i.e.
$\mu_n = \norm{\cdot}_{{\rm TV},2}$. The quality of completion, is
evaluated by the peak signal-to-noise ratio (PSNR) and structural
similarity index (SSIM) as in
\cite{Yokota16-smoothCP-tensor-completion}. High PSNR and SSIM
indicate a good completion. For both methods, we take a fix rank
$R = 50$ and use a random initialization. For each image, each mask and each model,
we use the oracle selection of $\alpha$ based on the SSIM. As observed in
\Cref{fig:image-alpha}, there is an optimal value of $\alpha$ which
should neither be too small nor too large. This value is not
necessarily the same for both algorithms. We also observe that, near
the optimal value of $\alpha$, the gradient-based method gives a
better SSIM score than the HALS method. This is also highlighted by
the completed images displayed in \Cref{fig:image}. Note also that the
gradient-based method has much lower TPI than the HALS one.
\begin{figure}[!h]
  \centering
  \includegraphics[width=0.49\textwidth]{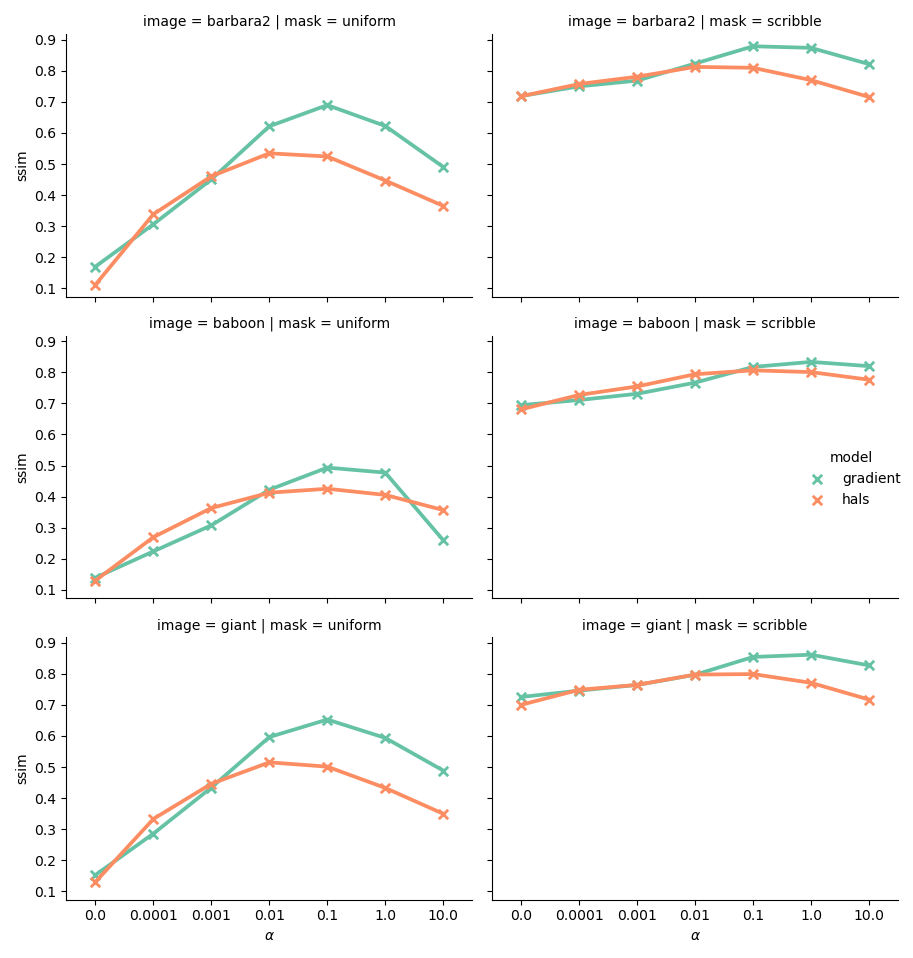}
  \caption{Evolution of the SSIM with $\alpha$. Rows represent  the three images and  columns represent the types of mask.\label{fig:image-alpha}}
\end{figure}

\begin{figure*}[h]
  \subfloat[Uniformly missing pixels.\label{fig:image-all}]{%
    \includegraphics[width=0.5\textwidth]{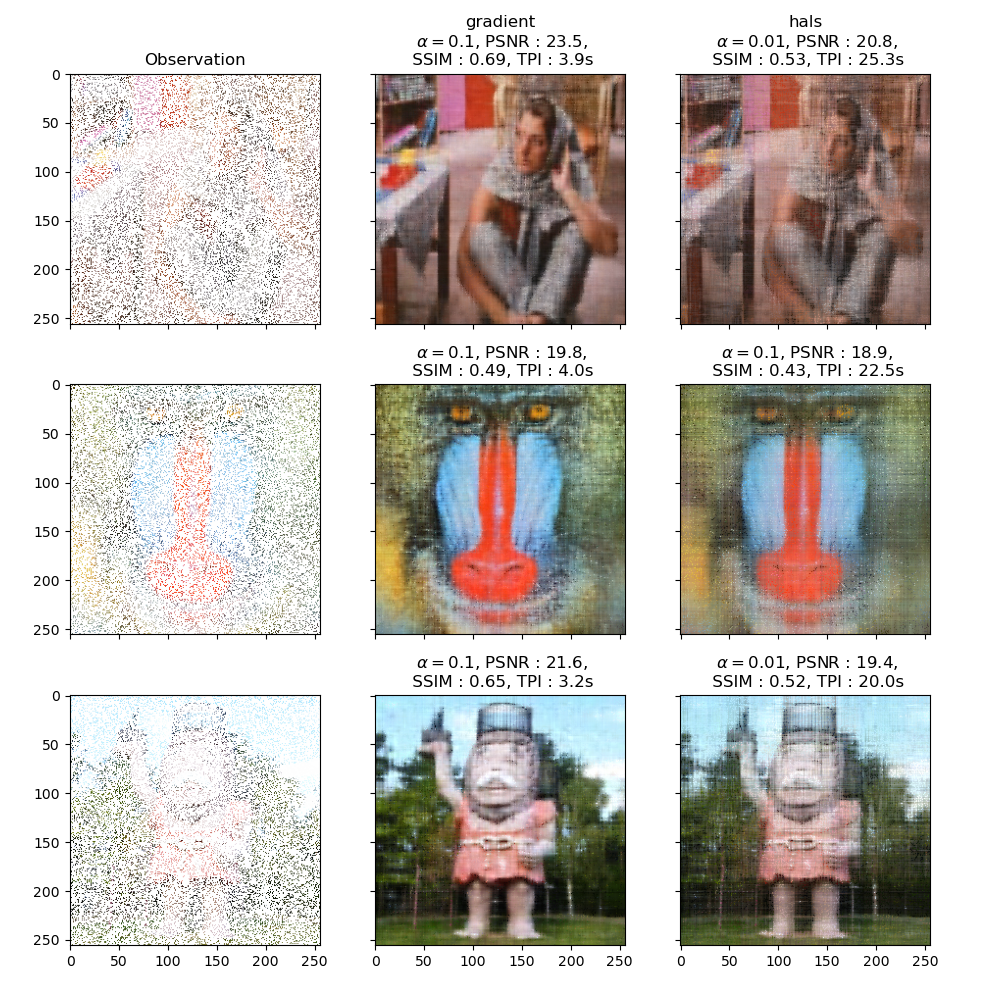}
  }
  \hfill
  \subfloat[Masked pixels. \label{fig:image-scribble}]{%
    \includegraphics[width=0.5\textwidth]{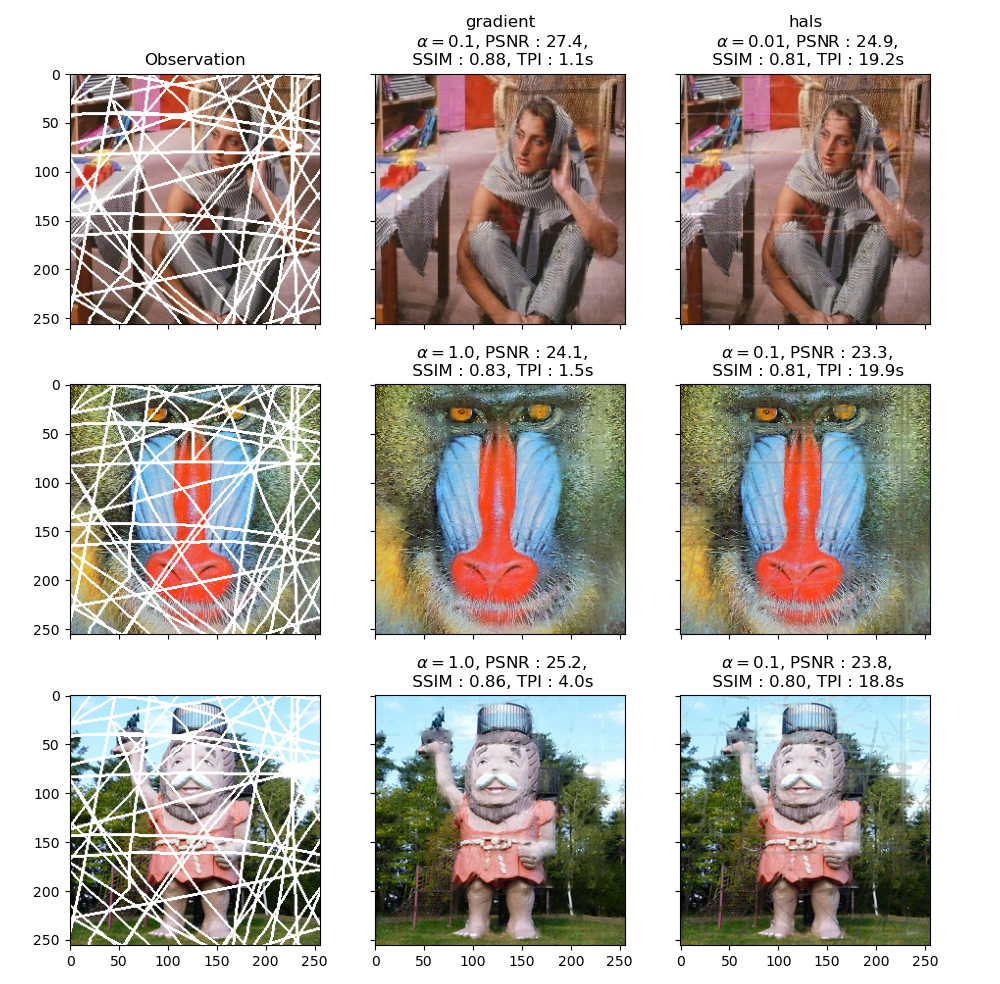}
  }
 \caption{Image completion comparison with $\alpha$ giving the best SSIM.}\label{fig:image}
\end{figure*}

\section{Conclusion}
In this contribution, we extended previous results on the existence of
a global minimum for the non-negative tensor factorization problem to
the case where some entries are missing. We showed that, under
non-restrictive assumptions on the observed entries, adding a penalty
to the quadratic loss ensures the existence of a global minimum. We
proposed two formulations of the problem: a normalized one, which is
solved using a HALS algorithm, and an unnormalized one, which is solved
using a gradient-based method. The experimental study illustrates the
advantages of the gradient-based method in computing time and in
reconstruction error.

\section{Proofs}\label{sec:proofs}
\subsection{Preliminary results}
In this section, we provide preliminary results which are necessary to prove \Cref{thm:coercivity-penalized}.
To simplify the notation, let us define, for all $\btheta = (\blambda,\bA^{(1:N)}) \in \Theta := \rset_+^R \times \prod_{n=1}^N \left(\mathbb{S}_{\nu_n}^+\right)^R$,
\begin{equation}\label{eq:def-g-loss}
g_\bW(\btheta) := \norm{\bW \oast \left(\sum_{r=1}^R \lambda_r \ba_r^{(1)} \circ \cdots \circ \ba_r^{(N)}\right)}_F \;,
\end{equation}
and, for all $\ba^{(1:N)} \in \prod_{n=1}^N \mathbb{S}_{\nu_n}^+$,
\begin{equation}\label{eq:def-f-loss}
 h_{\bW}(\ba^{(1:N)}) := \norm{\bW \oast \left(\ba^{(1)}\circ\cdots\circ\ba^{(N)}\right)}_F^2 \; .
\end{equation}
Then the following proposition holds. 
\begin{proposition}\label{prop:wsntf-restricted} Let
  $\cA = \prod_{n=1}^N \cA_n^R$ be such that for all
  $n \in \iseg{1,N}$, $\cA_n \subset \mathbb{S}_{\nu_n}^+$. Denote by
  $\overline{\cA_n}$ the closure of $\cA_n$ in
  $\mathbb{S}_{\nu_n}^+$. Then the three following assertions are
  equivalent.
  \begin{enumerate}[label=(\roman*)]
  \item\label{itm:compact-levelset-L} The function $\tilde{L}_\bW$ is coercive on $\rset_+^R \times \cA$.
  \item\label{itm:coercivity-g} The function $g_\bW$ is coercive on $\rset_+^R \times \cA$.
  \item\label{itm:condition-solution} For all $\ba^{(1:N)} \in \prod_{n=1}^N \overline{\cA_n}$, there exists $\bi \in \{\bW > 0\}$ such that $\prod_{n=1}^N \ba_{i_n}^{(n)} > 0$. 
  \end{enumerate}
\end{proposition}
\begin{proof}
  The equivalence between \ref{itm:compact-levelset-L} and
  \ref{itm:coercivity-g} is a direct consequence of the two triangular
  inequalities.  Now, note that, by continuity of $h_\bW$,
  Assertion~\ref{itm:condition-solution} is equivalent to
\begin{equation}\label{eq:inf-W-A-stricly-pos}
  \inf
  \set{h_\bW(\ba^{(1:N)})}{\ba^{(1:N)}\in\prod_{n=1}^N \cA_n} > 0  \; . 
\end{equation}
Hence, it remains to prove that \ref{itm:coercivity-g} and \eqref{eq:inf-W-A-stricly-pos} are equivalent. 

\noindent{\bf Proof of \eqref{eq:inf-W-A-stricly-pos} $\Rightarrow$
  \ref{itm:coercivity-g}.} Assume that
Condition~\eqref{eq:inf-W-A-stricly-pos} holds and let
$(\btheta(m))_{m \in \nset} \in \left(\rset_+^R \times
  \cA\right)^\nset$ with $\norm{\btheta(m)}_{2} \xrightarrow[m\to +\infty]{} +\infty$ and, for all $m\in\nset$,
$\btheta(m) = (\blambda(m), \bA^{(1:N)}(m))$.
Then, since $\cA$ is bounded, we must have $\norm{\blambda(m)}_2 \xrightarrow[m\to +\infty]{} +\infty$ and,
using the fact that the entries of $\btheta(m)$ are all non-negative, we get
  \begin{align*}
    \left(g_{\bW}(\btheta(m))\right)^2
    &\geq \sum_{r=1}^R \left(\lambda_r(m)\right)^2 h_{\bW}(\ba_r^{(1:N)}(m)) \\
    &\geq c_{\inf} \norm{\blambda(m)}_2^2  \; ,
  \end{align*}
  where $c_{\inf}$ is the inf
  in~(\ref{eq:inf-W-A-stricly-pos}). Hence, under Condition~\eqref{eq:inf-W-A-stricly-pos},
  $g_{\bW}(\btheta(m))$ diverges to $+\infty$ as $m\to +\infty$ and
  Assertion~\ref{itm:coercivity-g} follows.

  \noindent{\bf Proof of \ref{itm:coercivity-g} $\Rightarrow$
    \eqref{eq:inf-W-A-stricly-pos}.} Assume that
  Condition~\eqref{eq:inf-W-A-stricly-pos} does not hold and let us
  show that $\tilde{g}_{\bW}$ is not coercive on
  $\rset_+^R \times \cA$ by constructing a sequence
  $(\btheta(m))_{m\in\nset} \in \left(\rset_+^R \times
    \cA\right)^\nset$ such that
  $\norm{\btheta(m)} \xrightarrow[m\to +\infty]{} +\infty$ and
  $g_\bW(\btheta(m)) \xrightarrow[m\to +\infty]{} 0$.  Let us set, for
  all $m\in\nset$, $\btheta(m) := (\blambda(m), \bA^{(1:N)}(m))$ with
  $\blambda(m) = [m,0,\cdots,0]$ and, for all $n\in\nset$,
  $\bA^{(n)}(m) = [\ba^{(n)}(m), \ba_2^{(n)},\cdots, \ba_R^{(n)}]$,
  where $\ba_2^{(n)},\cdots,\ba_R^{(n)}$ are any elements of $\cA_n$
  and $\ba^{(1:N)}(m) \in \prod_{n=1}^N \cA_n$ is such that
  $h_\bW(\ba^{(1:N)}(m)) \leq 2^{-m}$. Note that such $\ba^{(1:N)}(m)$
  exists because Condition~\eqref{eq:inf-W-A-stricly-pos} does not
  hold. With this construction, we get
  $\norm{\btheta(m)} \xrightarrow[m\to +\infty]{} +\infty$ and
  $\left(g_\bW(\btheta(m))\right)^2 = m^2 h_\bW(\ba^{(1:N)}(m)) \leq
  m^2 2^{-m} \xrightarrow[m\to+\infty]{} 0$, which concludes the
  proof.
\end{proof}
For the second result, we consider specific subsets
\begin{align}\label{eq:def-A-rho-C}
&  \cA_{n}(\bfrho,C) := \set{\ba \in \mathbb{S}_{\nu_n}^+}{\mu_n(\ba) \leq \rho_n C} \;,\\
\label{eq:def-theta-rho-C}
&  \Theta(\bfrho,C) := \rset_+^R \times \prod_{n=1}^N \left(\cA_{n}(\bfrho,C)\right)^R \; ,
\end{align}
where $C \geq 0$ and $\bfrho \in \oseg{0,+\infty}^N$. Note that, for
all $n$ such that $\rho_n = +\infty$, we have
$\cA_{n}(\bfrho,C)= \mathbb{S}_{\nu_n}^+$.  Let us define
\begin{equation}\label{eq:def-C-W-rho}
  C_\bW(\bfrho)
  :=\inf \set{C \geq 0}{
    \begin{tabular}{l}
    $\tilde{L}_{\bW}$ is not  coercive \\ on $\Theta(\bfrho,C)$
    \end{tabular}    
  }  \; .
\end{equation}
The goal of the remaining of this section it to derive an explicit
formulation of $C_\bW(\bfrho)$. To this end, we need to introduce the
following sets and constants.  For any $n\in\iseg{1,N}$ and
$\cJ \subset \iseg{1,I_n}$, define
\begin{align*}
&\cA_{n}(\cJ) := \set{\ba \in \mathbb{S}_{\nu_n}^+}{\forall j \in \cJ, a_j = 0}\;,\\
&m_{n}(\cJ) := \inf\set{\mu_n(\ba)}{\ba\in\cA_{n}(\cJ)} \;,  
\end{align*}
with the convention $\inf(\emptyset) = +\infty$. Moreover, for any
$\cI' \subset \cI$ and $n \in \iseg{1,N}$, we define the projection of
$\cI'$ onto the $n$-th coordinate as
$\pi_n(\cI') := \set{i_n}{\bi \in \cI'}$ and denote by
$$
\mathscr{J}(\cI') := \set{\cJ_{1:N}}{
  \begin{split}
    &\forall n \in \iseg{1,N}\,,\,\cJ_n\subseteq\pi_n(\cI')\,,\\
&  \forall \bi \in \cI'\,,\, \exists  n \in \iseg{1,N}, i_n \in \cJ_n    
\end{split}
}
$$
the class of sequences of sets $\cJ_{1:N} = (\cJ_1,\cdots,\cJ_N)$, where each $\cJ_n$
is a subset of all the $n$-th entries of the vector indices in $\cI'$ and such
that each vector index $\bi$ in $\cI'$ has at least one entry, say the $n$-th,
present in the corresponding $\cJ_n$.

Finally, we define, for all $\bfrho \in \oseg{0,+\infty}^N$ and $\cI'
\subset \cI$,
\begin{align}\label{eq:def-Cpq}
C(\bfrho,\cI') := \min\set{c(\bfrho,\cJ_{1:N})}{\cJ_{1:N} \in
                                  \mathscr{J}(\cI')} \; ,
\end{align}
where
$$
    c(\bfrho,\cJ_{1:N}) := \begin{cases}
      \infty & \text{if $\exists n\,,\,\cA_n(\cJ_n)=\emptyset$,} \\
      \displaystyle\max_{n
        \in
        \iseg{1,N}}
      \frac{m_{n}(\cJ_n)}{\rho_n} & \text{otherwise.}
    \end{cases}
    $$
The first case in the definition of $c(\bfrho,\cJ_{1:N})$ amounts to
use the convention $\inf(\emptyset)/\rho = +\infty$ for any $\rho\in(0,\infty]$.
We have the following results.
\begin{lemma}\label{lem:ineq-mpq}
  Let $n \in \iseg{1,N}$, and $\cJ \subset \iseg{1,I_n}$, the following assertions hold.
  \begin{enumerate}[label=(\roman*)]
  \item\label{itm:m-finite} We have $m_n(\cJ) = +\infty$ if and only if $\cJ = \iseg{1,I_n}$.
  \item\label{itm:m-null} If $\mu_n$ satisfies~\ref{item:assumption-seminorms}, we have
    $m_n(\cJ) = 0$ if and only if $\cJ = \emptyset$. 
  \end{enumerate}
\end{lemma}
\begin{proof}
  Assertion~\ref{itm:m-finite} follows from the equivalence between
  $m_n(\cJ) = +\infty$ and $\cA_n(\cJ) = \emptyset$, which is itself
  equivalent to $\cJ = \iseg{1,I_n}$. For Assertion~\ref{itm:m-null},
  recall that $\psi_n=\infty$ means that $\mu_n$ is a semi-norm which is
  not a norm on the positive cone, which implies
  $m_n(\emptyset) = 0$ thus showing the ``if'' implication. For the
  ``only if'' implication,   note that $m_n(\cJ) > 0$ if and only if
  $$
  \sup\set{\nu_n(\ba)}{\ba \in\mathbb{B}^+_{\mu_n}(\ba)_,,\,\forall j\in \cJ\,,\,\ba_j = 0}<\infty\;.
  $$
  Hence, if $\psi'_n<\infty$, this happens whenever $\cJ \neq \emptyset$.
\end{proof}

\begin{lemma}\label{lem:condition-C-notnull}
 Let $\cI' \subset \cI$ and $\bfrho \in (0, +\infty]^N$. Then the following assertions hold.
  \begin{enumerate}[label=(\roman*)]
  \item\label{itm:C-finite} We have $C(\bfrho, \cI') = +\infty$ if and only if  $\cI' = \cI$.
  \item\label{itm:C-not-null} If $\mu_n$ satisfies~\ref{item:assumption-seminorms} for all $n\in\{\bfrho< +\infty\}$, then $C(\bfrho,\cI') > 0$ if and only if $\cI \setminus \cI'$ contains no $\{\bfrho = +\infty\}$-cylinder.
  \end{enumerate}
\end{lemma}
\begin{proof}
  Using Assertion~\ref{itm:m-finite} of \Cref{lem:ineq-mpq}, we get that 
  \begin{align*}
    C&(\bfrho, \cI') < +\infty \\
    &\Leftrightarrow  \exists\, \cJ_{1:N} \in \mathscr{J}(\cI'), \forall n \in \iseg{1,N}, m_{n}(\cJ_n) < +\infty\\
    &\Leftrightarrow  \exists\, \cJ_{1:N} \in \mathscr{J}(\cI'), \forall n \in \iseg{1,N}, \cJ_n \neq \iseg{1,I_n}\\
    &\Leftrightarrow  \exists\, \cJ_{1:N} \in \mathscr{J}(\cI'), \exists \bi \in \cI, \forall n \in \iseg{1,N}, i_n \notin \cJ_n \\
    &\Leftrightarrow  \mathscr{J}(\cI') \cap (\mathscr{J}(\cI))^c  \neq \emptyset  \; .
  \end{align*}
  Hence, to conclude the proof of Assertion~\ref{itm:C-finite}, we need to show that $\mathscr{J}(\cI') \cap (\mathscr{J}(\cI))^c  \neq \emptyset$ if and only if $\cI' \neq \cI$. The ``only if'' implication is straightforward by contraposition. For the ``if'' implication, assume that $\cI' \neq \cI$ and take $\bi \in \cI \setminus \cI'$. Then it is easily seen that, taking $\cJ_n = \pi_n(\cI') \setminus \{i_n\}$ for all $n\in\iseg{1,N}$, we get $\cJ_{1:N} \in \mathscr{J}(\cI') \cap (\mathscr{J}(\cI))^c$ which is therefore non-empty. This concludes the proof of Assertion~\ref{itm:C-finite}.
  
  For Assertion~\ref{itm:C-not-null}, note that $C(\bfrho,\cI') = 0$ is equivalent to
  $$
  \exists\, \cJ_{1:N} \in \mathscr{J}(\cI'),\, \forall n \in \iseg{1,N},\, \bfrho_n^{-1} m_{n}(\cJ_n) = 0 \; ,
  $$
  which, by \Cref{lem:ineq-mpq} and \ref{item:assumption-seminorms},
  is in turn equivalent to 
  \begin{equation}\label{eq:equiv-C-null} 
    \exists\, \cJ_{1:N} \in \mathscr{J}(\cI'),\,
      \begin{cases}
        \forall n \in \{\bfrho = +\infty\}, \, \cJ_n \neq \iseg{1,I_n}\;,\\
        \forall n \in \{\bfrho < +\infty\}, \, \cJ_n = \emptyset\;.\\
      \end{cases} 
    \end{equation}
    We now prove that \eqref{eq:equiv-C-null} holds if and only if there exists $(j_n)_{n\in\{\bfrho=+\infty\}} \in \prod_{n\in\{\bfrho=+\infty\}} \iseg{1,I_n}$ such that $\cC((j_n)_{n\in\{\bfrho=+\infty\}}) \subset \cI \setminus \cI'$. First, assume that  \eqref{eq:equiv-C-null} holds and take $j_n \in \iseg{1,I_n}\setminus \{\cJ_n\}$ for all $n\in\{\bfrho=+\infty\}$. Assume that $\cC((j_n)_{n\in\{\bfrho=+\infty\}}) \not\subset \cI \setminus \cI'$. This means that there exists $\bi \in \cC((j_n)_{n\in\{\bfrho=+\infty\}}) \cap \cI'$. Then, by definition of $\mathscr{J}(\cI')$, there exists $n\in\{\bfrho=+\infty\}$ such that $j_n = i_n \in \cJ_n$ which contradicts the fact that $j_n \in \iseg{1,I_n}\setminus \{\cJ_n\}$.  Hence $\cC((j_n)_{n\in\{\bfrho=+\infty\}}) \subset \cI \setminus \cI'$ thus proving the ``only if'' implication. For the ``if'' implication, assume that there exists $(j_n)_{n\in\{\bfrho=+\infty\}} \in \prod_{n\in\{\bfrho=+\infty\}} \iseg{1,I_n}$ such that $\cC((j_n)_{n\in\{\bfrho=+\infty\}}) \subset \cI \setminus \cI'$. Then, we get \eqref{eq:equiv-C-null} by taking $\cJ_n = \emptyset$ for $n\in\{\bfrho<+\infty\}$ and $\cJ_n = \iseg{1,I_n}\setminus\{j_n\}$ for $n\in\{\bfrho=+\infty\}$. 
\end{proof}

\begin{lemma}\label{lem:def-C-with-m}
  For all $\bfrho \in \oseg{0,+\infty}^N$, we have
  \begin{equation}\label{eq:def-C-with-m}
    C_\bW(\bfrho) 
    = C(\bfrho,\{\bW > 0\}) \; ,
  \end{equation}
\end{lemma}
\begin{proof}
  The proof relies on the identity 
    \begin{multline}\label{eq:invert-inf-in-c}
      c(\bfrho, \cJ_{1:N}) = \\ 
      \inf \set{\max_{n\in\iseg{1,N}}\rho_n^{-1} \mu_n(\ba^{(n)})}{\forall n, \ba^{(n)} \in \cA_{n}(\cJ_n)}\; .
    \end{multline}    
    To prove \eqref{eq:invert-inf-in-c}, first note that, if there
    exists $n\in\iseg{1,I_n}$, such that $\cJ_n = \iseg{1,I_n}$, then
    the two terms of \eqref{eq:invert-inf-in-c} are equal to
    $+\infty$.  We now assume that, for all $n\in\iseg{1,I_n}$,
    $\cJ_n \neq \iseg{1,I_n}$. Then the inequality $\leq$ of
    \eqref{eq:invert-inf-in-c} is a straightforward consequence of the
    definition of $c(\bfrho, \cJ_{1:N})$. Let us now show that there
    exists $\ba^{(1:N)} \in \prod_{n=1}^N \cA_{n}(\cJ_n)$ such that
    $c(\bfrho,\cJ_{1:N}) = \max_{n\in\iseg{1,N}}\rho_n^{-1}
    \mu_n(\ba^{(n)})$. It suffices to take, for $n\in\iseg{1,N}$,
    $\ba^{(n)} \in \cA_{n}(\cJ_n)$ such that
    $\mu_n(\ba^{(n)}) = m_{n}(\cJ_n)$. Such $\ba^{(n)}$ exists because
    $\cA_{n}(\cJ_n)$ is compact and $\mu_n$ is continuous. This
    concludes the proof of \eqref{eq:invert-inf-in-c}.

    Now, to prove \eqref{eq:def-C-with-m}, we show that, for all $C \geq 0$, $\tilde{L}_\bW$ is not coercive on $\Theta(\bfrho,C)$ if and only if $C \geq C(\bfrho,\{\bW > 0\})$. For the ``only if'' implication, assume that $C$ is such that $\tilde{L}_\bW$ is not coercive on $\Theta(\bfrho,C)$. Then, from \Cref{prop:wsntf-restricted} and closeness of the $\cA_{n}(\bfrho,C)$'s defined in \eqref{eq:def-A-rho-C}, we get that there exists $\ba^{(1:N)} \in \prod_{n=1}^N \cA_{n}(\bfrho,C)$ such that, for all $\bi\in\{\bW > 0\}$, there exists $n\in\iseg{1,N}$ such that  $a_{i_n}^{(n)} = 0$. Using this observation, we construct $\cJ_{1:N} \in \mathscr{J}(\{\bW > 0\})$ by the following procedure. Start with $\cJ_n = \emptyset$ for all $n\in\iseg{1,N}$ and then, for each $\bi \in \{\bW > 0\}$, select one of the $n$'s such that $a_{i_n}^{(n)} = 0$ and put $i_n$ in $\cJ_n$. With this construction, we have $\ba^{(1:N)} \in \prod_{n=1}^N \cA_{n}(\cJ_n) \cap \cA_n(\bfrho,C)$ and therefore 
$$
C \geq \max_{n\in\iseg{1,N}} \rho_n^{-1} \mu_n(\ba^{(n)})
\geq c(\bfrho, \cJ_{1:N})
\geq C(\bfrho,\{\bW > 0\}) \; ,
$$
where the first inequality comes from \eqref{eq:def-A-rho-C}, the second from \eqref{eq:invert-inf-in-c} and the last from \eqref{eq:def-Cpq}.

For the ``if'' implication, let us take $C \geq C(\bfrho,\{\bW > 0\})$ and show that Assertion~\ref{itm:condition-solution} of \Cref{prop:wsntf-restricted} does not hold. Note that we can take $\bW$ such that $\{\bW > 0\} \neq \cI$ because the other case is straightforward. Then, by definition of $C(\bfrho,\{\bW > 0\})$, we know that there exists $\cJ_{1:N} \in\mathscr{J}(\{\bW > 0\})$ such that $C(\bfrho,\{\bW > 0\}) = c(\bfrho,\cJ_{1:N})$. Moreover, since $\{\bW > 0\} \neq \cI$, Assertion~\ref{itm:m-finite} of \Cref{lem:ineq-mpq} gives that $c(\bfrho,\cJ_{1:N}) < +\infty$ and therefore we are in the case where the infimum in \eqref{eq:invert-inf-in-c} is reached.  Hence there exists $\ba^{(1:N)}\in \prod_{n=1}^N \cA_{n}(\cJ_n)$ such that $\max_{n\in\iseg{1,N}}\rho_n^{-1}\mu_n(\ba^{(n)}) = C(\bfrho,\{\bW > 0\})$. This gives that, for all $n\in\iseg{1,N}$, $\mu_n(\ba^{(n)}) \leq \rho_n C(\bfrho, \{\bW > 0\}) \leq \rho_n C$  and therefore  $\ba^{(n)} \in \cA_{n}(\bfrho,C)$, as defined in \eqref{eq:def-A-rho-C}. On the other hand, $\cJ_{1:N} \in \mathscr{J}(\{\bW > 0\})$ means that, for all $\bi \in\{\bW > 0\}$, there exists $n\in\iseg{1,N}$ such that $i_n \in \cJ_n$ and we get that $\ba_{i_n}^{(n)}= 0$ because $\ba^{(n)}\in\cA_{n}(\cJ_n)$. Hence Assertion~\ref{itm:condition-solution} of \Cref{prop:wsntf-restricted} does  not hold and the proof is concluded.
\end{proof}

\subsection{Proofs of \Cref{thm:coercivity-penalized} and \Cref{lem:smoothness-continuity-norms}}
\begin{proof}[Proof of \Cref{thm:coercivity-penalized}] In the proof,
  we use the functions $g_\bW$ and $h_\bW$ defined respectively in
  \eqref{eq:def-g-loss} and \eqref{eq:def-f-loss}. Note that, as in \Cref{prop:wsntf-restricted}, we have that
  $\tilde{f}_{\bW,\balpha}$ is coercive on $\Theta$ if and only if
  $g_\bW^2 + \tilde{\cP}_{\balpha}$ is coercive on $\Theta$.

\noindent {\bf Proof of \ref{itm:no-missing-cylinder} $\Rightarrow$ \ref{itm:frscp-compact-level-set}.} Assume that \ref{itm:no-missing-cylinder} holds and let us show that $g_\bW^2 + \tilde{\cP}_{\balpha}$ is coercive on $\Theta$. Take $\btheta = (\blambda,\bA^{(1:N)})\in \Theta$ and let $s\in\iseg{1,R}$ be such that $\lambda_s = \max_{r\in\iseg{1,R}}\lambda_r= \norm{\blambda}_\infty$. Then, since all the entries of $\btheta$ are non-negative, we have
  \begin{align*}
    (g_\bW&(\btheta))^2 + \tilde{\cP}_{\balpha}(\btheta) \\
    &\geq (\lambda_s^2 \wedge \lambda_s^d) \left( h_{\bW}(\ba_s^{(1:N)}) + \sum_{n=1}^N \alpha_n \mu_n^p(\ba_s^{(n)}))\right) \\
    &\geq \left[\norm{\blambda}_{\infty}^{d\wedge 2} - 1\right]_+ \eta \;,
  \end{align*}
  where
  $$
  \eta := \inf_{\ba^{(1:N)} \in \prod_{n=1}^N \mathbb{S}_{\nu_n}^+}\left(h_\bW(\ba^{(1:N)})+ \sum_{n=1}^N\alpha_n \mu_n^p(\ba^{(n)})\right)\; . 
  $$
  Hence, to prove that $g_\bW^2 + \tilde{\cP}_{\balpha}$ is coercive
  on $\Theta$, it suffices to prove that $\eta > 0$. Now let us set,
  for all $n\in\iseg{1,N}$, $\rho_n = \alpha_n^{-p}$ with the
  convention that $0^{-p} = +\infty$. Then, \Cref{lem:def-C-with-m}
  and Assertion~\ref{itm:C-not-null} of \Cref{lem:condition-C-notnull}
  give that $C_\bW(\bfrho) > 0$. Take
  $0 < C < C_{\bW}(\bfrho)$ and set
  $\cA := \prod_{n=1}^N \cA_{n}(\bfrho,C)$ with $\cA_{n}(\bfrho,C)$
  defined as in \eqref{eq:def-A-rho-C}. Then, by definition of
  $C_\bW(\bfrho)$ in~(\ref{eq:def-C-W-rho}) and Relation~\eqref{eq:inf-W-A-stricly-pos}, we have
  $\inf_\cA h_\bW > 0$. Moreover, for all $\ba^{(1:N)} \in \cA^c$,
  there exists $k\in\iseg{1,N}$ such that
  $\ba^{(k)} \notin \cA_{k}(\bfrho,C)$, i.e.
  $\alpha_k \mu^p(\ba_r^{(k)}) > C^p$ and
  $\sum_{n=1}^N \alpha_n \mu_n^p(\ba^{(n)}) \geq \alpha_k
  \mu_n^p(\ba_r^{(k)}) > C^p$. Hence, we get
  $h_\bW + \sum_{n=1}^N \alpha_n \mu_n^p \geq \left(\inf_\cA
    h_\bW\right) \1_\cA + C^p \1_{\cA^c}$ so that
  $\eta \geq \left(\inf_\cA h_\bW\right) \wedge C^p > 0$, thus
  concluding the proof.

  \noindent {\bf Proof of \ref{itm:frscp-compact-level-set}
    $\Rightarrow$ \ref{itm:no-missing-cylinder}.}  Let us assume that
  \ref{itm:no-missing-cylinder} does not hold and show that
  $g_\bw^2 + \tilde{\cP}_{\balpha}$ is not coercive on $\Theta$ by
  constructing a sequence $(\btheta(m))_{m\in\nset}\in\Theta^\nset$
  such that $\btheta(m) \xrightarrow[m\to +\infty]{} +\infty$ and
  $g_\bW^2(\btheta(m)) + \tilde{\cP}_{\balpha}(\btheta(m))
  \xrightarrow[m\to +\infty]{} 0$.  We set, for all
  $n\in\iseg{1,N}$, $\rho_n = \alpha_n^{-p}$ with the convention that
  $0^{-p} = +\infty$. Then, from Assertion~\ref{itm:C-not-null} of
  \Cref{lem:condition-C-notnull}, we get that $C_\bW(\bfrho) = 0$. By
  definition of $C_\bW(\bfrho)$ and
  Relation~\eqref{eq:inf-W-A-stricly-pos}, this gives that for all
  $m \in \nset$, $\inf_{\cA^{(m)}} h_\bW = 0$ with
  $\cA^{(m)} = \prod_{n=1}^N \cA_n(\bfrho,2^{-m/p})$. In particular,
  we can find $\ba^{(1:N)}(m) \in \cA^{(m)}$ such that
  $h_\bW(\ba^{(1:N)}(m)) \leq 2^{-m}$. Now, take, for all $m\in\nset$,
  $\btheta(m) = (\blambda(m), \bA^{(1:N)})$ with
  $\blambda(m) = [m,0,\cdots,0]$ and, for all $n\in\iseg{1,N}$,
  $\bA^{(n)}(m) = [\ba^{(n)}(m), \ba_2^{(n)},\cdots, \ba_R^{(n)}]$,
  where $\ba_2^{(n)}, \cdots, \ba_R^{(n)}$ are arbitrary elements of
  $\mathbb{S}_{\nu_n}^+$. In this case, we have
  $\norm{\btheta(m)}_2 \xrightarrow[m \to +\infty]{} +\infty$, but
  $\left(g_\bW(\btheta(m))\right)^2 +
  \tilde{\cP}_{\balpha}(\btheta(m)) \leq (m^2 + N m^d) 2^{-m}
  \xrightarrow[m\to+\infty]{} 0$, which concludes the proof.
\end{proof}

\begin{proof}[Proof of \Cref{lem:smoothness-continuity-norms}]
  The fact that none of the $\mu_n$'s in Points~\ref{itm:tv-p} and
  \ref{itm:spline} is a norm on the positive cone proves that 
  $\psi_n=\infty$ in both cases. We now show that $\psi'_n<\infty$
  holds in both cases. By equivalence of the norms, we can assume
  without loss of generality that $\nu_n = \norm{\cdot}_\infty$. For
  $\mu_n$ as in Case~\ref{itm:tv-p}, we can take $p = 1$
  without loss of generality. Then, for all $i,j\in\iseg{1,I_n}$, we
  have
  $\abs{a_j - a_i} \leq \sum_{k=j\wedge i}^{j\vee i - 1} \abs{a_{k+1} -
    a_k} \leq \norm{\ba}_{{\rm TV},1}$. This implies
  $\psi'_n<\infty$. Next, we take $\mu_n$ as in
  Case~\ref{itm:spline}. Let $\ba \in\mathbb{B}^+_n$ such that $a_j=0$
  for some $j\in\iseg{1,I_n}$.  Let
  $(\hat{a}_k)_{k\in\zset}$ be the Fourier coefficients of the
  corresponding Spline function $a$,
  $\hat{a}_k=\int_0^1a(u)\,\rme^{-2\rmi\pi ku}\;\rmd u$. Then, we have
  $\mu_n(\ba) = \varsqrt{\sum_{k\in\zset} \abs{\hat{a}_k}^2 (2\pi k)^2}\leq1$.
  Now, using the Cauchy-Schwarz inequality and the fact that
  $1/12=\sum_{k\in\zset^*} (2\pi k)^{-2}$, where $\zset^*$ is the set
  of non-zero relative integers, we get that
  $$
  \sum_{k\in\zset^*} \abs{\hat{a}_k} \leq \varsqrt{\sum_{k\in\zset^*} (2\pi k)^{-2}} \mu_n(\ba) \leq \frac{1}{2\sqrt{3}} \; . 
  $$
  In particular, we get $\sum_{k\in\zset}\abs{\hat{a}_k} < +\infty$ and therefore, for all $u\in [0,1]$, we have
  $a(u) = \sum_{k\in\zset} \hat{a}_k \rme^{2\rmi \pi k u}$. This implies 
  that
  \begin{align*}
    \norm{\ba}_{\infty}
    &=    \norm{\ba - a_j}_{\infty}\\
    &\leq \sup_{u\in [0,1]} \abs{a(u) - a(u_j)} \\
    &\leq \sum_{k\in\zset^*} \abs{\hat{a}_k} \sup_{u\in [0,1]} \abs{\rme^{2 \rmi \pi k u} - \rme^{2 \rmi \pi k u_j}} \\
    &\leq 2 \sum_{k\in\zset^*}  \abs{\hat{a}_k} \; .                       
  \end{align*}
Hence, with the previous inequality, $\psi'_n\leq \frac1{\sqrt{3}}$ in this case, which concludes the
proof.
\end{proof}



\begin{thebibliography}{10}
\providecommand{\url}[1]{#1}
\csname url@samestyle\endcsname
\providecommand{\newblock}{\relax}
\providecommand{\bibinfo}[2]{#2}
\providecommand{\BIBentrySTDinterwordspacing}{\spaceskip=0pt\relax}
\providecommand{\BIBentryALTinterwordstretchfactor}{4}
\providecommand{\BIBentryALTinterwordspacing}{\spaceskip=\fontdimen2\font plus
\BIBentryALTinterwordstretchfactor\fontdimen3\font minus
  \fontdimen4\font\relax}
\providecommand{\BIBforeignlanguage}[2]{{%
\expandafter\ifx\csname l@#1\endcsname\relax
\typeout{** WARNING: IEEEtran.bst: No hyphenation pattern has been}%
\typeout{** loaded for the language `#1'. Using the pattern for}%
\typeout{** the default language instead.}%
\else
\language=\csname l@#1\endcsname
\fi
#2}}
\providecommand{\BIBdecl}{\relax}
\BIBdecl

\bibitem{KoBa09}
T.~G. Kolda and B.~W. Bader, ``Tensor decompositions and applications,''
  \emph{SIAM Review}, vol.~51, no.~3, pp. 455--500, September 2009.

\bibitem{Cichocki-NTF}
A.~Cichocki, R.~Zdunek, A.~H. Phan, and S.~Amari, \emph{Nonnegative Matrix and
  Tensor Factorizations: Applications to Exploratory Multi-Way Data Analysis
  and Blind Source Separation}.\hskip 1em plus 0.5em minus 0.4em\relax Wiley
  Publishing, 2009.

\bibitem{Sidiropoulos17-tensor-sp-ml}
N.~D. Sidiropoulos, L.~De~Lathauwer, X.~Fu, K.~Huang, E.~E. Papalexakis, and
  C.~Faloutsos, ``Tensor decomposition for signal processing and machine
  learning,'' \emph{IEEE Transactions on Signal Processing}, vol.~65, no.~13,
  pp. 3551--3582, 2017.

\bibitem{Hillar-tensor-NP-Hard}
\BIBentryALTinterwordspacing
C.~J. Hillar and L.-H. Lim, ``Most tensor problems are np-hard,'' \emph{J.
  Acm}, vol.~60, no.~6, Nov. 2013. [Online]. Available:
  \url{https://doi.org/10.1145/2512329}
\BIBentrySTDinterwordspacing

\bibitem{Silva08-illposed-tensor}
\BIBentryALTinterwordspacing
V.~de~Silva and L.-H. Lim, ``Tensor rank and the ill-posedness of the best
  low-rank approximation problem,'' \emph{SIAM J. Matrix Anal. Appl.}, vol.~30,
  no.~3, pp. 1084--1127, 2008. [Online]. Available:
  \url{https://doi.org/10.1137/06066518X}
\BIBentrySTDinterwordspacing

\bibitem{Lim09-NTF}
\BIBentryALTinterwordspacing
L.-H. Lim and P.~Comon, ``Nonnegative approximations of nonnegative tensors,''
  \emph{Journal of Chemometrics}, vol.~23, no. 7‐8, pp. 432--441, 2009.
  [Online]. Available:
  \url{https://onlinelibrary.wiley.com/doi/abs/10.1002/cem.1244}
\BIBentrySTDinterwordspacing

\bibitem{Tikhonov-ill-posed}
A.~N. Tikhonov and V.~Y. Arsenin, \emph{Solutions of ill-posed problems}, ser.
  Scripta Series in Mathematics.\hskip 1em plus 0.5em minus 0.4em\relax V. H.
  Winston \& Sons, Washington, D.C.: John Wiley \& Sons, New York-Toronto,
  Ont.-London, 1977, translated from the Russian, Preface by translation editor
  Fritz John.

\bibitem{Vapnik-statistical-learning-theory}
V.~N. Vapnik, \emph{Statistical learning theory}, ser. Adaptive and Learning
  Systems for Signal Processing, Communications, and Control.\hskip 1em plus
  0.5em minus 0.4em\relax John Wiley \& Sons, Inc., New York, 1998, a
  Wiley-Interscience Publication.

\bibitem{lim2005optimal}
L.-H. Lim, ``Optimal solutions to non-negative parafac/multilinear nmf always
  exist,'' in \emph{Workshop on Tensor Decompositions and Applications, Centre
  International de rencontres Math{\'e}matiques, Luminy, France}, 2005.

\bibitem{Bro97-parafac}
\BIBentryALTinterwordspacing
R.~Bro, ``Parafac. tutorial and applications,'' \emph{Chemometrics and
  Intelligent Laboratory Systems}, vol.~38, no.~2, pp. 149--171, 1997.
  [Online]. Available:
  \url{https://www.sciencedirect.com/science/article/pii/S0169743997000324}
\BIBentrySTDinterwordspacing

\bibitem{ANDERSSON-Nway-toolbox}
\BIBentryALTinterwordspacing
C.~A. Andersson and R.~Bro, ``The n-way toolbox for matlab,''
  \emph{Chemometrics and Intelligent Laboratory Systems}, vol.~52, no.~1, pp.
  1--4, 2000. [Online]. Available:
  \url{https://www.sciencedirect.com/science/article/pii/S016974390000071X}
\BIBentrySTDinterwordspacing

\bibitem{TOMASI-parafac-missing-values}
\BIBentryALTinterwordspacing
G.~Tomasi and R.~Bro, ``Parafac and missing values,'' \emph{Chemometrics and
  Intelligent Laboratory Systems}, vol.~75, no.~2, pp. 163--180, 2005.
  [Online]. Available:
  \url{https://www.sciencedirect.com/science/article/pii/S0169743904001741}
\BIBentrySTDinterwordspacing

\bibitem{ACAR-tf-missing-data}
\BIBentryALTinterwordspacing
E.~Acar, D.~M. Dunlavy, T.~G. Kolda, and M.~M{\o}rup, ``Scalable tensor
  factorizations for incomplete data,'' \emph{Chemometrics and Intelligent
  Laboratory Systems}, vol. 106, no.~1, pp. 41--56, 2011, multiway and Multiset
  Data Analysis. [Online]. Available:
  \url{https://www.sciencedirect.com/science/article/pii/S0169743910001437}
\BIBentrySTDinterwordspacing

\bibitem{Xiong10-BPTF}
\BIBentryALTinterwordspacing
L.~Xiong, X.~Chen, T.-K. Huang, J.~Schneider, and J.~G. Carbonell,
  \emph{Temporal Collaborative Filtering with Bayesian Probabilistic Tensor
  Factorization}, 2010, pp. 211--222. [Online]. Available:
  \url{https://epubs.siam.org/doi/abs/10.1137/1.9781611972801.19}
\BIBentrySTDinterwordspacing

\bibitem{pmlr-v32-rai14}
\BIBentryALTinterwordspacing
P.~Rai, Y.~Wang, S.~Guo, G.~Chen, D.~Dunson, and L.~Carin, ``Scalable bayesian
  low-rank decomposition of incomplete multiway tensors,'' in \emph{Proceedings
  of the 31st International Conference on Machine Learning}, ser. Proceedings
  of Machine Learning Research, E.~P. Xing and T.~Jebara, Eds., vol.~32,
  no.~2.\hskip 1em plus 0.5em minus 0.4em\relax Bejing, China: Pmlr, 22--24 Jun
  2014, pp. 1800--1808. [Online]. Available:
  \url{https://proceedings.mlr.press/v32/rai14.html}
\BIBentrySTDinterwordspacing

\bibitem{Zhao16-bayesian-tensor}
\BIBentryALTinterwordspacing
Q.~Zhao, G.~Zhou, L.~Zhang, A.~Cichocki, and S.-I. Amari, ``Bayesian robust
  tensor factorization for incomplete multiway data,'' \emph{IEEE Trans. Neural
  Netw. Learn. Syst.}, vol.~27, no.~4, pp. 736--748, 2016. [Online]. Available:
  \url{https://doi.org/10.1109/TNNLS.2015.2423694}
\BIBentrySTDinterwordspacing

\bibitem{Song19-tensorcompletion}
\BIBentryALTinterwordspacing
Q.~Song, H.~Ge, J.~Caverlee, and X.~Hu, ``Tensor completion algorithms in big
  data analytics,'' \emph{ACM Trans. Knowl. Discov. Data}, vol.~13, no.~1, Jan.
  2019. [Online]. Available: \url{https://doi.org/10.1145/3278607}
\BIBentrySTDinterwordspacing

\bibitem{TIMMERMAN2002447}
\BIBentryALTinterwordspacing
M.~E. Timmerman and H.~A. Kiers, ``Three-way component analysis with smoothness
  constraints,'' \emph{Computational Statistics \& Data Analysis}, vol.~40,
  no.~3, pp. 447--470, 2002. [Online]. Available:
  \url{https://www.sciencedirect.com/science/article/pii/S0167947302000592}
\BIBentrySTDinterwordspacing

\bibitem{Reis02-parafac-spline}
\BIBentryALTinterwordspacing
M.~M. Reis and M.~M.~C. Ferreira, ``Parafac with splines: a case study,''
  \emph{Journal of Chemometrics}, vol.~16, no. 8‐10, pp. 444--450, 2002.
  [Online]. Available:
  \url{https://onlinelibrary.wiley.com/doi/abs/10.1002/cem.749}
\BIBentrySTDinterwordspacing

\bibitem{YOKOTA15-smoothCP}
\BIBentryALTinterwordspacing
T.~Yokota, R.~Zdunek, A.~Cichocki, and Y.~Yamashita, ``Smooth nonnegative
  matrix and tensor factorizations for robust multi-way data analysis,''
  \emph{Signal Processing}, vol. 113, pp. 234--249, 2015. [Online]. Available:
  \url{https://www.sciencedirect.com/science/article/pii/S0165168415000614}
\BIBentrySTDinterwordspacing

\bibitem{Yokota16-tensor-completion-smooth}
T.~Yokota and A.~Cichocki, ``Tensor completion via functional smooth component
  deflation,'' in \emph{2016 IEEE International Conference on Acoustics, Speech
  and Signal Processing (ICASSP)}, 2016, pp. 2514--2518.

\bibitem{Li17-low-rank-completion}
X.~Li, Y.~Ye, and X.~Xu, ``Low-rank tensor completion with total variation for
  visual data inpainting,'' in \emph{Proceedings of the Thirty-First AAAI
  Conference on Artificial Intelligence}, ser. Aaai'17.\hskip 1em plus 0.5em
  minus 0.4em\relax AAAI Press, 2017, p. 2210–2216.

\bibitem{Imaizumi17-tensor-smoothness}
\BIBentryALTinterwordspacing
M.~Imaizumi and K.~Hayashi, ``Tensor decomposition with smoothness,'' in
  \emph{Proceedings of the 34th International Conference on Machine Learning},
  ser. Proceedings of Machine Learning Research, D.~Precup and Y.~W. Teh, Eds.,
  vol.~70.\hskip 1em plus 0.5em minus 0.4em\relax International Convention
  Centre, Sydney, Australia: Pmlr, 06--11 Aug 2017, pp. 1597--1606. [Online].
  Available: \url{http://proceedings.mlr.press/v70/imaizumi17a.html}
\BIBentrySTDinterwordspacing

\bibitem{Sadowski18-imagecompletion}
T.~Sadowski and R.~Zdunek, ``Image completion with smooth nonnegative matrix
  factorization,'' in \emph{Artificial Intelligence and Soft Computing},
  L.~Rutkowski, R.~Scherer, M.~Korytkowski, W.~Pedrycz, R.~Tadeusiewicz, and
  J.~M. Zurada, Eds.\hskip 1em plus 0.5em minus 0.4em\relax Cham: Springer
  International Publishing, 2018, pp. 62--72.

\bibitem{durand21-SmoothNTF}
\BIBentryALTinterwordspacing
A.~Durand, F.~Roueff, J.-M. Jicquel, and N.~Paul, ``Smooth nonnegative tensor
  factorization for multi-sites electrical load monitoring,'' in
  \emph{{Eusipco}}, Dublin, Ireland, Aug. 2021. [Online]. Available:
  \url{https://hal.telecom-paris.fr/hal-03167498}
\BIBentrySTDinterwordspacing

\bibitem{Henriet19NILM}
U.~Henriet, S.and~\c{S}im\c{s}ekli, S.~Dos~Santos, B.~Fuentes, and G.~Richard,
  ``Independent-variation matrix factorization with application to energy
  disaggregation,'' \emph{IEEE Signal Processing Letters}, vol.~26, no.~11, pp.
  1643--1647, 2019.

\bibitem{gousseau01}
\BIBentryALTinterwordspacing
Y.~Gousseau and J.-M. Morel, ``Are natural images of bounded variation?''
  \emph{SIAM Journal on Mathematical Analysis}, vol.~33, no.~3, pp. 634--648,
  2001. [Online]. Available: \url{https://doi.org/10.1137/S0036141000371150}
\BIBentrySTDinterwordspacing

\bibitem{Yokota16-smoothCP-tensor-completion}
T.~Yokota, Q.~Zhao, and A.~Cichocki, ``Smooth parafac decomposition for tensor
  completion,'' \emph{IEEE Transactions on Signal Processing}, vol.~64, no.~20,
  pp. 5423--5436, 2016.

\bibitem{Zdunek14-splineNMF}
R.~Zdunek, A.~Cichocki, and T.~Yokota, ``B-spline smoothing of feature vectors
  in nonnegative matrix factorization,'' in \emph{Artificial Intelligence and
  Soft Computing}, L.~Rutkowski, M.~Korytkowski, R.~Scherer, R.~Tadeusiewicz,
  L.~A. Zadeh, and J.~M. Zurada, Eds.\hskip 1em plus 0.5em minus 0.4em\relax
  Cham: Springer International Publishing, 2014, pp. 72--81.

\bibitem{Amini17-functional-cp}
\BIBentryALTinterwordspacing
A.~A. Amini, E.~Levina, and K.~A. Shedden, ``Structured regression models for
  high-dimensional spatial spectroscopy data,'' \emph{Electronic Journal of
  Statistics}, vol.~11, no.~2, pp. 4151 -- 4178, 2017. [Online]. Available:
  \url{https://doi.org/10.1214/17-EJS1301}
\BIBentrySTDinterwordspacing

\bibitem{HAUTECOEUR2020256}
C.~Hautecoeur and F.~Glineur, ``Nonnegative matrix factorization over
  continuous signals using parametrizable functions,'' \emph{Neurocomputing},
  vol. 416, pp. 256--265, 2020.

\end{thebibliography}
\end{document}